\documentclass[12pt]{amsart}
\usepackage{amssymb}
\usepackage{color}
\usepackage{amsmath,epic,curves,amscd}
\usepackage[english]{babel}
\usepackage{graphicx}
\usepackage{comment}
\usepackage{appendix}
\usepackage{mathdots}
\usepackage[all]{xy}
\newtheorem{claim}{}[section]
\newtheorem{theorem}[claim]{Theorem}

\newtheorem{lemma}[claim]{Lemma}
\newtheorem{proposition}[claim]{Proposition}
\newtheorem{corollary}[claim]{Corollary}

\theoremstyle{remark}

\renewenvironment{proof}{\noindent{\it Proof. \hskip0pt}}
                      {$\square$\par\medskip}

\begin{document}
\baselineskip 6.0 truemm
\parindent 1.5 true pc

\newcommand\lan{\langle}
\newcommand\ran{\rangle}
\newcommand\tr{{\text{\rm Tr}}\,}
\newcommand\ot{\otimes}
\newcommand\ol{\overline}
\newcommand\join{\vee}
\newcommand\meet{\wedge}
\renewcommand\ker{{\text{\rm Ker}}\,}
\newcommand\image{{\text{\rm Im}}\,}
\newcommand\id{{\text{\rm id}}}
\newcommand\tp{{\text{\rm tp}}}
\newcommand\pr{\prime}
\newcommand\e{\epsilon}
\newcommand\la{\lambda}
\newcommand\inte{{\text{\rm int}}\,}
\newcommand\ttt{{\text{\rm t}}}
\newcommand\spa{{\text{\rm span}}\,}
\newcommand\conv{{\text{\rm conv}}\,}
\newcommand\rank{\ {\text{\rm rank of}}\ }
\newcommand\re{{\text{\rm Re}}\,}
\newcommand\ppt{\mathbb T}
\newcommand\rk{{\text{\rm rank}}\,}
\newcommand\SN{{\text{\rm SN}}\,}
\newcommand\SR{{\text{\rm SR}}\,}
\newcommand\HA{{\mathcal H}_A}
\newcommand\HB{{\mathcal H}_B}
\newcommand\HC{{\mathcal H}_C}
\newcommand\CI{{\mathcal I}}
\newcommand{\bra}[1]{\langle{#1}|}
\newcommand{\ket}[1]{|{#1}\rangle}
\newcommand\cl{\mathcal}
\newcommand\idd{{\text{\rm id}}}
\newcommand\OMAX{{\text{\rm OMAX}}}
\newcommand\OMIN{{\text{\rm OMIN}}}
\newcommand\diag{{\text{\rm Diag}}\,}
\newcommand\calI{{\mathcal I}}
\newcommand\bfi{{\bf i}}
\newcommand\bfj{{\bf j}}
\newcommand\bfk{{\bf k}}
\newcommand\bfl{{\bf l}}
\newcommand\bfp{{\bf p}}
\newcommand\bfq{{\bf q}}
\newcommand\bfzero{{\bf 0}}
\newcommand\bfone{{\bf 1}}
\newcommand\im{{\mathcal R}}
\newcommand\ha{{\frac 12}}
\newcommand{\para}{\mathbin{\!/\mkern-3mu/\!}}

\title{The role of phases in detecting three-qubit entanglement}

\author{Kyung Hoon Han and Seung-Hyeok Kye}
\address{Department of Mathematics, The University of Suwon, Gyeonggi-do 445-743, Korea}
\email{kyunghoon.han at gmail.com}
\address{Department of Mathematics and Institute of Mathematics, Seoul National University, Seoul 151-742, Korea}
\email{kye at snu.ac.kr}
\thanks{KHH and SHK were partially supported by NRF-2012R1A1A1012190 and 2015R1D1A1A02061612, respectively.}



\begin{abstract}
We propose separability criteria for three-qubit states in terms of
diagonal and anti-diagonal entries to detect entanglement with
positive partial transposes.
We report that the phases, that is, the angular parts of anti-diagonal entries, play a crucial role in determining whether a given three-qubit state is separable or entangled, and they must obey even an identity for separability in some cases.
These criteria are strong enough to detect PPT (positive partial transpose) entanglement with nonzero volume. In several cases when
all the entries are zero except for diagonal and anti-diagonal entries,
we characterize separability using phases. These include the cases
when anti-diagonal entries of such states share a common magnitude,
and when ranks are less than or equal to six. We also compute the lengths of rank six cases,
and find three-qubit separable states with lengths $8$
whose maximum ranks of partial transposes are $7$.
\end{abstract}

\maketitle

%%%%%%%%%%%%%%%%%%%%%%%%%%%%%%%%%%%%%%%%%%%%%%%%%%%%%%%%%%%%%%%%%%%%%%%%%%%%%%%%%%%%%%%%%%%%%%%%%%%%%%%%%%%%%%%%%%%%%%%%%%%%%%%%%%%%%%%%
%%%%%%%%%%%%%%%%%%%%%%%%%%%%%%%%%%%%%%%%%%%%%%%%%%%%%%%%%%%%%%%%%%%%%%%%%%%%%%%%%%%%%%%%%%%%%%%%%%%%%%%%%%%%%%%%%%%%%%%%%%%%%%%%%%%%%%%%
%%%%%%%%%%%%%%%%%%%%%%%%%%%%%%%%%%%%%%%%%%%%%%%%%%%%%%%%%%%%%%%%%%%%%%%%%%%%%%%%%%%%%%%%%%%%%%%%%%%%%%%%%%%%%%%%%%%%%%%%%%%%%%%%%%%%%%%%
%%%%%%%%%%%%%%%%%%%%%%%%%%%%%%%%%%%%%%%%%%%%%%%%%%%%%%%%%%%%%%%%%%%%%%%%%%%%%%%%%%%%%%%%%%%%%%%%%%%%%%%%%%%%%%%%%%%%%%%%%%%%%%%%%%%%%%%%
%%%%%%%%%%%%%%%%%%%%%%%%%%%%%%%%%%%%%%%%%%%%%%%%%%%%%%%%%%%%%%%%%%%%%%%%%%%%%%%%%%%%%%%%%%%%%%%%%%%%%%%%%%%%%%%%%%%%%%%%%%%%%%%%%%%%%%%%
%%%%%%%%%%%%%%%%%%%%%%%%%%%%%%%%%%%%%%%%%%%%%%%%%%%%%%%%%%%%%%%%%%%%%%%%%%%%%%%%%%%%%%%%%%%%%%%%%%%%%%%%%%%%%%%%%%%%%%%%%%%%%%%%%%%%%%%%
%%%%%%%%%%%%%%%%%%%%%%%%%%%%%%%%%%%%%%%%%%%%%%%%%%%%%%%%%%%%%%%%%%%%%%%%%%%%%%%%%%%%%%%%%%%%%%%%%%%%%%%%%%%%%%%%%%%%%%%%%%%%%%%%%%%%%%%%
%%%%%%%%%%%%%%%%%%%%%%%%%%%%%%%%%%%%%%%%%%%%%%%%%%%%%%%%%%%%%%%%%%%%%%%%%%%%%%%%%%%%%%%%%%%%%%%%%%%%%%%%%%%%%%%%%%%%%%%%%%%%%%%%%%%%%%%%
%%%%%%%%%%%%%%%%%%%%%%%%%%%%%%%%%%%%%%%%%%%%%%%%%%%%%%%%%%%%%%%%%%%%%%%%%%%%%%%%%%%%%%%%%%%%%%%%%%%%%%%%%%%%%%%%%%%%%%%%%%%%%%%%%%%%%%%%
%%%%%%%%%%%%%%%%%%%%%%%%%%%%%%%%%%%%%%%%%%%%%%%%%%%%%%%%%%%%%%%%%%%%%%%%%%%%%%%%%%%%%%%%%%%%%%%%%%%%%%%%%%%%%%%%%%%%%%%%%%%%%%%%%%%%%%%%
\section{Introduction}

The notion of entanglement is a unique phenomenon in quantum
physics, and is now considered as one of the main resources in
various fields of current quantum information and computation
theory, like quantum cryptography and quantum teleportation. See
survey articles \cite{{guhne_survey},{horo-survey}} for general
aspects on the topics. In cases of multipartite systems, there are
several kinds of entanglement as it was classified in
\cite{{abls},{dur_multi},{dur}}, and it is important to find
separability criteria to distinguish entanglement from separability.
Positivity of partial transposes is a simple but powerful criterion
\cite{{choi-ppt},{peres}}.

Some of other criteria are to test quite simple relations between
diagonal and anti-diagonal entries of states. See
\cite{{guhne_pla_2011},{guhne10},{Rafsanjani}} for example. This
approach is very successful to detect kinds of multi-qubit
entanglement arising from bi-separability or full bi-separability,
and some of them actually characterize those separability when a
given state has zero entries except for diagonal and anti-diagonal
entries \cite{{guhne10},{han_kye_EW},{Rafsanjani}}. The main purpose
of this note is to give criteria for (full) separability of three
qubit states in terms of diagonal and anti-diagonal entries. These
criteria tell us that the (full) separability of three-qubit states
depends heavily on the phases, that is, the angular parts of
anti-diagonal entries. Furthermore, they detect PPT (positive partial transpose) entanglement
with nonzero volume. We recall that a state is said to be separable,
or fully separable, if it is a convex combination of pure product
states, or equivalently, that of product states.

By anti-diagonal entries of an $n\times n$ matrix $[a_{i,j}]$, we
mean $a_{i,n-i+1}$ for $i=1,2,\dots,n$.
States with zero entries except for diagonal and anti-diagonal
entries are usually called {\sf X}-shaped states, or {\sf X}-states,
in short. Those states arise naturally in quantum information theory
in various aspects. See \cite{{abls},{mendo},{rau},{vin10},{wein10},{yu}}
for example. Notable examples include Greenberger-Horne-Zeilinger
diagonal states, which are mixtures of GHZ states with noises. We
note that the {\sf X}-part of a three-qubit separable state is again
separable \cite{han_kye_GHZ}, and so any necessary criteria for
separability of {\sf X}-shaped states still work for arbitrary three
qubit states in terms of diagonal and anti-diagonal entries.
Very little is known for full separability of three-qubit X-states,
even though we now have a complete characterization of
bi-separability and full bi-separability of arbitrary multi-qubit {\sf X}-states \cite{{han_kye_EW},{Rafsanjani}}. Only recently, separability of three-qubit GHZ diagonal
states has been completely characterized by the authors
\cite{han_kye_GHZ}, complimenting earlier partial results in
\cite{{guhne_pla_2011},{kay_2011}}. We note that anti-diagonal
entries of GHZ diagonal states are real numbers, and so it is very
natural to ask what happens when the anti-diagonal part has complex
entries.

Three-qubit states are considered as $8\times 8$ matrices, by
the identification $M_8=M_2\otimes M_2\otimes M_2$ with the lexicographic order
of indices in the tensor product. Therefore,
a three-qubit {\sf X}-shaped Hermitian matrix is of the form
\begin{equation}
X(a,b,c)= \left(
\begin{matrix}
a_1 &&&&&&& c_1\\
& a_2 &&&&& c_2 & \\
&& a_3 &&& c_3 &&\\
&&& a_4&c_4 &&&\\
&&& \bar c_4& b_4&&&\\
&& \bar c_3 &&& b_3 &&\\
& \bar c_2 &&&&& b_2 &\\
\bar c_1 &&&&&&& b_1
\end{matrix}
\right),
\end{equation}
for $a,b\in\mathbb R^4$ and $c\in\mathbb C^4$.
We also denote by $\theta_i$ the phase of $c_i=r_ie^{{\rm i}\theta_i}$ throughout this note.
For a given state $\varrho$ whose {\sf X}-part is given by $X(a,b,c)$,
we define the following three numbers:
\begin{equation}
\begin{aligned}
\Delta_\varrho&=\min\{\sqrt{a_1b_1}, \sqrt{a_2b_2}, \sqrt{a_3b_3}, \sqrt{a_4b_4}, \sqrt[4]{a_1b_2b_3a_4}, \sqrt[4]{b_1a_2a_3b_4}\},\\
R_\varrho&=\max\{ |c_1|, |c_2|, |c_3|, |c_4| \},\\
r_\varrho&=\min\{ |c_1|, |c_2|, |c_3|, |c_4| \}.
\end{aligned}
\end{equation}

We note that if a three-qubit state $\varrho$ with the {\sf
X}-part $X(a,b,c)$ is separable then it satisfies the inequality
\begin{equation}\label{ppt-diag}
\Delta_\varrho\ge R_\varrho.
\end{equation}
The inequality $\sqrt{a_ib_i}\ge R_\varrho$ comes from the
positivity of partial transposes, as it was also observed for
general multi-qubit states \cite{han_kye_EW}. The others appear in
\cite{guhne_pla_2011}. This criterion gives us a restriction on the
maximum of magnitudes of anti-diagonal entries for separable states.
We consider here the minimum of the anti-diagonal magnitudes, and
show that a separable state $\varrho$ must satisfy the inequality:
\begin{equation}\label{Phi}
\Delta_\varrho\ge r_\varrho \sqrt{1+|\sin \phi_\varrho/2|},
\end{equation}
which depends on the {\sl phase difference} $\phi_\varrho$ of $\varrho$ defined by
\begin{equation}
\phi_\varrho=(\theta_1+\theta_4)-(\theta_2+\theta_3) \mod 2\pi.
\end{equation}
If the phase difference is nonzero then our criterion (\ref{Phi})
detects entanglement with PPT property. Actually, we will see that
the set of all PPT entanglement detected by (\ref{Phi}) has nonzero
volume.

For given fixed diagonal parts and
magnitudes of anti-diagonals, our criterion also gives rise to a restriction
on the phase difference for a separable state. Especially, if
$\Delta_\varrho=R_\varrho=r_\varrho$, then separable states must
obey the {\sl phase identity}:
\begin{equation}
\theta_1+\theta_4=\theta_2+\theta_3 \mod 2\pi
\end{equation}
which can be easily observed
for pure product states \cite{guhne_pla_2011}.
For  non-diagonal {\sf X}-states with rank $\le 6$, we will also see that
the phase identity is a necessary condition for separability,
with which we characterize separability of
them in terms of entries. If a three-qubit {\sf X}-state $\varrho$ shares
a common magnitude, that is, $R_\varrho=r_\varrho$,
then our criterion (\ref{Phi}) characterize separability.

After we prove the criterion (\ref{Phi}) in the next section, we
apply this result in Section 3 to give a complete characterization
of separability for three-qubit {\sf X}-states with rank four. This
will be used in Section 4 to show that the criterion (\ref{Phi})
gives us a complete characterization of separability for {\sf
X}-states with common anti-diagonal magnitudes. In Section 5, we
will show that separable {\sf X}-states with rank $\le 6$ must
satisfy the phase identity, and characterize their separability.
We compute in Section 6 the lengths of separable {\sf X}-states of rank six.
In some cases, the length exceeds the maximum rank of partial transposes.

%%%%%%%%%%%%%%%%%%%%%%%%%%%%%%%%%%%%%%%%%%%%%%%%%%%%%%%%%%%%%%%%%%%%%%%%%%%%%%%%%%%%%%%%%%%%%%%%%%%%%%%%%%%%%%%%%%
%%%%%%%%%%%%%%%%%%%%%%%%%%%%%%%%%%%%%%%%%%%%%%%%%%%%%%%%%%%%%%%%%%%%%%%%%%%%%%%%%%%%%%%%%%%%%%%%%%%%%%%%%%%%%%%%%%
%%%%%%%%%%%%%%%%%%%%%%%%%%%%%%%%%%%%%%%%%%%%%%%%%%%%%%%%%%%%%%%%%%%%%%%%%%%%%%%%%%%%%%%%%%%%%%%%%%%%%%%%%%%%%%%%%%
%%%%%%%%%%%%%%%%%%%%%%%%%%%%%%%%%%%%%%%%%%%%%%%%%%%%%%%%%%%%%%%%%%%%%%%%%%%%%%%%%%%%%%%%%%%%%%%%%%%%%%%%%%%%%%%%%%
%%%%%%%%%%%%%%%%%%%%%%%%%%%%%%%%%%%%%%%%%%%%%%%%%%%%%%%%%%%%%%%%%%%%%%%%%%%%%%%%%%%%%%%%%%%%%%%%%%%%%%%%%%%%%%%%%%
\section{Separability criterion with anti-diagonal phases}

In order to justify the criterion (\ref{Phi}), we begin with the
separability criterion by G\" uhne \cite{{guhne_pla_2011}} which was
simplified by the authors \cite{han_kye_GHZ}. It was shown that
every three-qubit separable state $\varrho$ with the {\sf X}-part
$X(a,b,c)$ satisfies the inequality
\begin{equation}\label{guhne}
|{\rm Re} \left(z_1 c_1 + z_2 c_2 + z_3 c_3 + z_4 \bar c_4 \right)|
\le\Delta_\varrho \max_\tau\left(|z_1 e^{{\rm i}\tau}+z_4|+ |z_2
e^{{\rm i}\tau}+\bar z_3|\right)
\end{equation}
for each $z\in\mathbb C^4$.
Note that the number in the right side of (\ref{guhne}) appear in the
characterization of {\sf X}-shaped three-qubit entanglement witnesses
\cite{han_kye_tri}, which correspond to positive bi-linear maps
between $2\times 2$ matrices \cite{kye_3qb_EW}.
In order to get a preliminary criterion, we introduce the number
\begin{equation}
A_\varrho =\frac 1{2\sqrt 2}\max_\theta\left(|c_1 e^{{\rm i}\theta}
+c_2| + |c_3e^{{\rm i}\theta}-c_4|\right)
\end{equation}
which is determined by the anti-diagonal parts.

For a three-qubit state $\varrho$, we denote by $\varrho^{\Gamma_A}$
the partial transpose with respect to the system $A$, and
$\varrho^{\Gamma_B}$, $\varrho^{\Gamma_C}$ similarly.
Then we have
$$
\begin{aligned}
X(a,b,(c_1,c_2,c_3,c_4))^{\Gamma_A}&=X(a,b,(\bar c_4,\bar c_3,\bar c_2,\bar c_1)),\\
X(a,b,(c_1,c_2,c_3,c_4))^{\Gamma_B}&=X(a,b,(c_3,c_4,c_1,c_2)),\\
X(a,b,(c_1,c_2,c_3,c_4))^{\Gamma_C}&=X(a,b,(c_2,c_1,c_4,c_3)).
\end{aligned}
$$
Therefore, we see that the number $A_\varrho$ is invariant under all the kinds of
partial transposes.

\begin{proposition}
If $\varrho$ is a three-qubit separable state with its {\sf X}-part $X(a,b,c)$,
then $\varrho$ satisfies the inequality
\begin{equation}\label{theta-delta}
\Delta_\varrho\ge A_\varrho.
\end{equation}
\end{proposition}

\begin{proof}
For a given $\theta$, define $\phi = -\arg (c_1 e^{{\rm i}\theta} +
c_2)$ and $\psi = -\arg(c_3 e^{{\rm i} \theta}-c_4)$, and take
$z=(e^{{\rm i}(\theta+\phi)}, e^{{\rm i} \phi}, e^{{\rm i}(\theta+\psi)}, -e^{-{\rm i}\psi})\in\mathbb C^4$
in the the inequality (\ref{guhne}). We have ${\rm Re}(z_4\bar c_4)={\rm Re}(\bar z_4 c_4)$ and
$$
\begin{aligned}
c_1 z_1 + c_2 z_2 + c_3 z_3 + c_4 \bar z_4
& = (c_1 e^{{\rm i}\theta} + c_2)e^{{\rm i}\phi} + (c_3 e^{{\rm i}\theta} - c_4)e^{{\rm i}\psi} \\
& = |c_1 e^{{\rm i}\theta} +c_2| + |c_3e^{{\rm i}\theta}-c_4|
\end{aligned}
$$
in the left hand side. In the right hand side of (\ref{guhne}), we see that
$$
| z_1 e^{{\rm i} \tau} + z_4|+| z_2 e^{{\rm i} \tau} + \bar z_3|
=|e^{{\rm i}(\theta + \phi + \psi + \tau)} -1| + |e^{{\rm i}(\phi + \theta + \psi + \tau )} +1|
$$
has the maximum $2\sqrt{2}$ through the variable $\tau$.
\end{proof}

Now, we consider the case when the anti-diagonal entries share a common
magnitude, say $R=|c_i|$ for each $i=1,2,3,4$. In this case, the
number $ A_\varrho$ has a natural geometric interpretation. To see
this, put
$$
\begin{aligned}
T(\theta): &=|e^{{\rm i}\theta_1}e^{{\rm i}\theta}+e^{{\rm
i}\theta_2}|
+|e^{{\rm i}\theta_3}e^{{\rm i}\theta}-e^{{\rm i}\theta_4}|\\
&=|e^{{\rm i}\theta}-e^{{\rm i}(\theta_2-\theta_1+\pi)}|+|e^{{\rm
i}\theta}-e^{{\rm i}(\theta_4-\theta_3)}|.
\end{aligned}
$$
Then the maximum of $T(\theta)$ occurs when the three points
$e^{{\rm i}(\theta_2-\theta_1+\pi)}$, $e^{{\rm i}(\theta_4-\theta_3)}$ and $e^{{\rm i}\theta}$
on the complex plane make an isosceles
triangle, and so we have $2\sqrt 2\le
\max_\theta T(\theta) \le 4$. See FIGURE 1. Since $
A_\varrho=R\max_\theta T(\theta)/2\sqrt 2$, we have
$$
R\le A_\varrho\le \sqrt 2\, R.
$$
Note that $ A_\varrho= R$ if and only if
$\max_\theta T(\theta)=2\sqrt 2$ if and only if two angles
$\theta_2-\theta_1+\pi$ and $\theta_4-\theta_3$ are antipodal if and
only if  $\theta_1+\theta_4=\theta_2+\theta_3$. We also have
$A_\varrho= \sqrt 2\, R$ if and only if $\theta_1+\theta_4$ and
$\theta_2+\theta_3$ are antipodal.

\setlength{\unitlength}{1 mm}
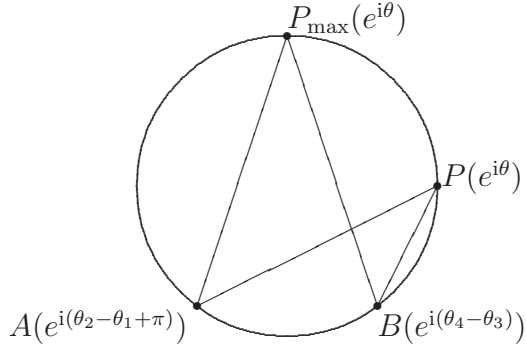
\begin{figure}
\setlength{\unitlength}{1 mm}
\begin{center}
\begin{picture}(50,50)
  \qbezier(25.000,0.000)(33.284,0.000)
          (39.142,5.858)
  \qbezier(39.142,5.858)(45.000,11.716)
          (45.000,20.000)
  \qbezier(45.000,20.000)(45.000,28.284)
          (39.142,34.142)
  \qbezier(39.142,34.142)(33.284,40.000)
          (25.000,40.000)
  \qbezier(25.000,40.000)(16.716,40.000)
          (10.858,34.142)
  \qbezier(10.858,34.142)( 5.000,28.284)
          ( 5.000,20.000)
  \qbezier( 5.000,20.000)( 5.000,11.716)
          (10.858,5.858)
  \qbezier(10.858,5.858)(16.716,0.000)
          (25.000,0.000)
\drawline(13,4)(25,40)
\drawline(37,4)(25,40)
\drawline(13,4)(45,20)
\drawline(37,4)(45,20)
\put(13,4){\circle*{1}}
\put(25,40){\circle*{1}}
\put(37,4){\circle*{1}}
\put(45,20){\circle*{1}}
\put(25,41){$P_{\max}(e^{{\rm i}\theta})$}
\put(45.5,20){$P(e^{{\rm i}\theta})$}
\put(-12,0){$A(e^{{\rm i}(\theta_2-\theta_1+\pi)})$}
\put(37,0){$B(e^{{\rm i}(\theta_4-\theta_3)})$}
\end{picture}
\end{center}
\caption{For two fix points $A$ and $B$ on the circle, the length
$\overline{AP}+\overline{PB}$ takes the maximum when the three
points $A,B$ and $P$ make an isosceles triangle. Furthermore, the
number $\max_P(\overline{AP}+\overline{PB})$ becomes largest when
$A$ and $B$ coincide, and smallest when $A$ and $B$ are antipodal.}
\end{figure}

Now, we proceed to find the maximum of $T(\theta)$. By the relation
$$
\max_\theta T(\theta)
=\max_\theta T(\theta-\theta_3+\theta_4)
=\max_\theta\{|e^{{\rm i}\theta} - e^{{\rm i}(-\phi_\varrho+\pi)}|+|e^{{\rm i}\theta}-1|\},
$$
we see that the maximum occurs when three points $1, e^{{\rm
i}(-\phi_\varrho+\pi)}$ and $e^{{\rm i}\theta}$ make an isosceles
triangle, that is, $\theta=-\phi_\varrho/2\pm  \pi/ 2$. Therefore,
we see that
$$
\begin{aligned}
\max_\theta T(\theta)
&=
\max_{\pm} \{|e^{{\rm i}(-\phi_\varrho/2\pm \pi /2)}
   - e^{{\rm i}(-\phi_\varrho+\pi)}|+|e^{{\rm i}(-\phi_\varrho/2\pm \pi /2)}-1|\},\\
&=
\max_\pm\, 2|e^{{\rm i}(-\phi_\varrho/2\pm \pi /2)}-1|\\
&=
\max_\pm 2\sqrt{2-2\cos(-\phi_\varrho/2 \pm \pi/2)}\\
&=
2\sqrt 2\max_\pm\sqrt{1\pm\sin \phi_\varrho/2}
=2\sqrt2\sqrt{1+|\sin\phi_\varrho/2|},
\end{aligned}
$$
whenever $R=|c_i|$ for each $i=1,2,3,4$.
Therefore, we have the following:

\begin{lemma}\label{T_theta}
%If $|c_1|=|c_2|=|c_3|=|c_4|=R$, then we have
%$$
%A_\varrho=\frac R{2\sqrt 2}\, \max_\theta T(\theta)=R\sqrt{1+|\sin\phi_\varrho/2|}.
%$$
%We have $\max_\theta T(\theta)=2\sqrt{2}\sqrt{1+|\sin{\theta_1+\theta_4-\theta_2-\theta_3 \over 2}|}$.
%Therefore, if $|c_1|=|c_2|=|c_3|=|c_4|=R$, then $A_\varrho = R \sqrt{1+|\sin \phi_\varrho / 2}$.
We have $\max_\theta T(\theta)=2\sqrt{2}\sqrt{1+|\sin{\phi_\varrho / 2}|}$.
Therefore, if $|c_1|=|c_2|=|c_3|=|c_4|=R$, then $A_\varrho = R \sqrt{1+|\sin \phi_\varrho / 2}|$.
\end{lemma}

Considering the general cases with
arbitrary magnitudes of anti-diagonal entries, we get the main criterion
{\rm (\ref{Phi})}. In fact, we show that the inequality
(\ref{theta-delta}) implies the inequality (\ref{Phi}).
Note that the number $|\sin\phi_\varrho/2|$ in the criterion is
invariant under three kinds of partial transposes of $\varrho$.

\begin{theorem}
If $\varrho$ is a three-qubit separable state with the {\sf X}-part $X(a,b,c)$,
then the inequality {\rm (\ref{Phi})} holds.
\end{theorem}

\begin{proof}
Let $c_i=r_i e^{{\rm i} \theta_i}$.
When $r_2 \le r_1$, we have
$$
|c_1 e^{{\rm i}\theta} + c_2|
=r_2 \left|{r_1 \over r_2} e^{{\rm i}(\theta+\theta_1-\theta_2)} + 1\right|
\ge r_2 |e^{{\rm i}(\theta+\theta_1-\theta_2)} + 1|
=r_2 |e^{{\rm i}\theta_1} e^{{\rm i}\theta}+ e^{{\rm i}\theta_2}|.
$$
In case of $r_1 \le r_2$, applying the above yields
$$
|c_1 e^{{\rm i}\theta} + c_2|
=|c_1 + c_2 e^{-{\rm i}\theta}|
\ge r_1 |e^{{\rm i}\theta_1} + e^{{\rm i}\theta_2} e^{-{\rm i}\theta}|
=r_1 |e^{{\rm i}\theta_1} e^{{\rm i}\theta}+ e^{{\rm i}\theta_2}|.
$$
It follows that
$|c_1 e^{{\rm i}\theta} + c_2|
\ge \min\{r_1,r_2\}|e^{{\rm i}\theta_1} e^{{\rm i}\theta}+ e^{{\rm i}\theta_2}|$,
which yields
$$
|c_1 e^{{\rm i}\theta} +c_2| + |c_3e^{{\rm i}\theta}-c_4|
\ge r_\varrho T(\theta).
$$
Therefore, we have
$A_\varrho \ge r_\varrho T(\theta)/2\sqrt 2$ for every $\theta$. By Lemma \ref{T_theta}, we have
\begin{equation}\label{yyyyy}
 A_\varrho\ge r_\varrho\, \sqrt{1+|\sin \phi_\varrho/2|}.
\end{equation}
The result follows from the criterion (\ref{theta-delta}).
\end{proof}

\begin{corollary}\label{phase}
Suppose that the diagonal and anti-diagonal entries of a three-qubit state $\varrho$
satisfy the relation $\Delta_\varrho=R_\varrho=r_\varrho$.
If $\varrho$ is separable, then it obeys the phase identity.
\end{corollary}

If the anti-diagonal entries of separable $\varrho$ share a common magnitude and it {\sf X}-part
is singular, then $\Delta_\varrho=R_\varrho=r_\varrho$, and so we have the following:

\begin{corollary}\label{angle-cor}
Suppose that the anti-diagonal entries of a three-qubit state $\varrho$ share a common magnitude
and its {\sf X}-part is singular. If $\varrho$ is separable, then it obeys the phase identity.
\end{corollary}

Corollary \ref{phase} will be useful to characterize separability of rank four {\sf X}-states in the next section.
In order to compare two criteria (\ref{ppt-diag}) and (\ref{Phi}),
we fix the diagonal parts $a,b$ and the {\sl phase part}
$(e^{{\rm i}\theta_1},e^{{\rm i}\theta_2},e^{{\rm i}\theta_3},e^{{\rm i}\theta_4})$,
and consider the four dimensional convex body
$\mathbb S$ consisting of $(r_1,r_2,r_3,r_4)\in\mathbb R^4$ so that the
{\sf X}-state $X(a,b,c)$ is separable. The criterion (\ref{ppt-diag})
tells us that $\mathbb S$ is sitting in the cube with the width
$\Delta_\varrho$. On the other hand, we see that the region $\mathbb
S$ is located in the union of strips with the width
$\Delta_\varrho/\sqrt{1+|\sin \phi_\varrho/2|}$ by (\ref{Phi}). See
FIGURE 2.

\begin{figure}
\setlength{\unitlength}{.8 mm}
\begin{center}
\begin{picture}(50,50)
\drawline(0,0)(0,50) \drawline(0,0)(50,0) \drawline(25,0)(25,25)
\drawline(0,25)(25,25) \drawline(20,20)(50,20)
\drawline(20,20)(20,50) \put(20,20){\circle*{2}}
\put(21,21){\circle*{0.5}} \put(21,22){\circle*{0.5}}
\put(21,23){\circle*{0.5}} \put(21,24){\circle*{0.5}}
\put(22,21){\circle*{0.5}} \put(22,22){\circle*{0.5}}
\put(22,23){\circle*{0.5}} \put(22,24){\circle*{0.5}}
\put(23,21){\circle*{0.5}} \put(23,22){\circle*{0.5}}
\put(23,23){\circle*{0.5}} \put(23,24){\circle*{0.5}}
\put(24,21){\circle*{0.5}} \put(24,22){\circle*{0.5}}
\put(24,23){\circle*{0.5}} \put(24,24){\circle*{0.5}}
\end{picture}
\end{center}
\caption{The square in the picture represents the four dimensional
cube determined by the criterion (\ref{ppt-diag}), and the union of
strips is the region determined by (\ref{Phi}). So, the region
$\mathbb S$ located in the intersection of these regions, and the
dotted part is PPT entanglement detected by the criterion
(\ref{Phi}). It will be shown in Section 4 that the \lq corner point\rq\ of the strip, shown by a big dot,
is on the boundary of the separability region $\mathbb S$.}
\end{figure}
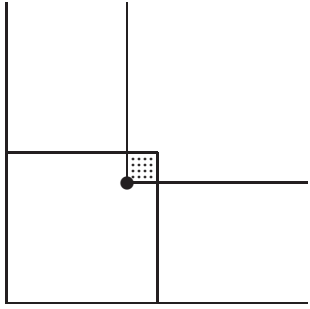

Now, we take an {\sf X}-state $\varrho=X(a,b,c)$ satisfying the
strict inequality in (\ref{ppt-diag}) but violating (\ref{Phi}).
Then we have $|c_i|< \sqrt{a_jb_j}$ for every $i,j=1,2,3,4$.
This means that  all the partial
transposes of $\varrho$ have full ranks, and so we see that $\varrho$ is an interior
point of the set of all three-qubit PPT states. Therefore, we conclude that PPT
entanglement detected by (\ref{Phi}) has nonzero volume.

%%%%%%%%%%%%%%%%%%%%%%%%%%%%%%%%%%%%%%%%%%%%%%%%%%%%%%%%%%%%%%%%%%%%%%%%%%%%%%%%%%%%%%%%%%%%%%%%%%%%%%%%%%%%%%%%%%
%%%%%%%%%%%%%%%%%%%%%%%%%%%%%%%%%%%%%%%%%%%%%%%%%%%%%%%%%%%%%%%%%%%%%%%%%%%%%%%%%%%%%%%%%%%%%%%%%%%%%%%%%%%%%%%%%%
%%%%%%%%%%%%%%%%%%%%%%%%%%%%%%%%%%%%%%%%%%%%%%%%%%%%%%%%%%%%%%%%%%%%%%%%%%%%%%%%%%%%%%%%%%%%%%%%%%%%%%%%%%%%%%%%%%
%%%%%%%%%%%%%%%%%%%%%%%%%%%%%%%%%%%%%%%%%%%%%%%%%%%%%%%%%%%%%%%%%%%%%%%%%%%%%%%%%%%%%%%%%%%%%%%%%%%%%%%%%%%%%%%%%%
%%%%%%%%%%%%%%%%%%%%%%%%%%%%%%%%%%%%%%%%%%%%%%%%%%%%%%%%%%%%%%%%%%%%%%%%%%%%%%%%%%%%%%%%%%%%%%%%%%%%%%%%%%%%%%%%%%
\section{Separable {\sf X}-states with rank four}

In this section, we characterize the separability of three-qubit {\sf X}-states with rank four
in terms of their diagonal and anti-diagonal entries.
It was shown in \cite{han_kye_GHZ} that the {\sf X}-part $\varrho_X$
of a pure product state $\varrho=|\xi\rangle\langle \xi|$ with
$|\xi\rangle=|x\rangle\otimes |y\rangle\otimes |z\rangle$ is given by the average of four pure product states:
\begin{equation}\label{average}
\omega_X=\frac 14 \sum_{k=0}^3 |\xi(k)\ran\lan\xi(k)|,
\end{equation}
where $|\xi(k)\ran$ is given by
\begin{equation}\label{decom-rank1}
\begin{aligned}
|\xi(0)\ran=&|x_+ \ran \ot |y_+ \ran \ot |z_+ \ran,\\
|\xi(1)\ran=&|x_+ \ran \ot |y_- \ran \ot |z_- \ran,\\
|\xi(2)\ran=&|x_- \ran \ot |y_+ \ran \ot |z_- \ran,\\
|\xi(3)\ran=&|x_- \ran \ot |y_- \ran \ot |z_+ \ran,
\end{aligned}
\end{equation}
with the notations $|x_\pm\ran =(x_0,\pm x_1)^\ttt$ and $|y_\pm\ran$,
$|z_\pm\ran$ similarly. This simple observation was used to see
in \cite{han_kye_GHZ} that the {\sf X}-part of a separable state
is again separable, and the {\sf X}-part of an entanglement witness is again an entanglement witness.
Actually, the {\sf X}-parts of the pure product states of these four vectors coincide.
If the {\sf X}-part of $\varrho$ is given by $X(a,b,c)$, then we have
$$
\sqrt{a_i b_i} = |x_0 y_0 z_0 x_1 y_1 z_1 | = |c_j|, \qquad i,j=1,2,3,4.
$$
Therefore, we see that the {\sf X}-part of $\varrho=|\xi\rangle\langle \xi|$ is of rank $\le 4$,
and is of rank four if it is not diagonal.

\begin{proposition}\label{extreme}
Every three-qubit separable {\sf X}-state is a convex combination of
non-diagonal separable {\sf X}-states of rank four and diagonal states.
\end{proposition}

\begin{proof}
Suppose that $\varrho$ is a separable {\sf X}-state, and write $\varrho=\sum_i\lambda_i|\xi_i\ran\lan \xi_i|$
with product vectors $|\xi_i\ran$.
Take the {\sf X}-part in both side. If $|\xi\ran$ has no zero entry then the {\sf X}-part of $|\xi\ran\lan\xi|$
is a non-diagonal separable {\sf X}-states of rank four. On the other hand,
if $|\xi\ran$ has a zero entry then the {\sf X}-part of $|\xi\ran\lan\xi|$
is a diagonal state.
\end{proof}

We proceed to show that the decomposition (\ref{average}) is unique.
It was shown in \cite{ha-kye-sep-face} that a separable state
$\varrho=\sum_i\lambda_i|\xi_i\ran\lan \xi_i|$ into pure product
states has a unique decomposition whenever the following two conditions
are satisfied:
\begin{enumerate}
\item[(A)]
the family $\{|\xi_i\ran\lan \xi_i|\}$ of pure product states is linearly independent
in the real vector space of all Hermitian matrices,
\item[(B)]
a product vector which belongs to the span of $\{|\xi_i\ran\}$ is parallel to one of $|\xi_i\ran$.
\end{enumerate}
The proof is same for multipartite cases.

\begin{proposition}\label{xpart-1}
The {\sf X}-part of a pure product state is a separable state
with a unique decomposition into pure product states.
\end{proposition}

\begin{proof}
It remains to show the uniqueness of decomposition. It also suffices to prove this
when the {\sf X}-part of $\varrho=|\xi\ran\lan\xi|$ is non-diagonal, or equivalently all the entries of
$|\xi\ran=|x\ran\otimes |y\ran\otimes |z\ran$ are nonzero.
We denote by $|x^\perp\ran, |y^\perp\ran$ and $|z^\perp\ran$ the orthogonal vectors
to $|x\ran, |y\ran$ and $|z\ran$, respectively.
If $\sum_{k=0}^3 a_k|\xi(k)\ran\lan\xi(k)|=0$, then we have
$$
0=\sum_{k=0}^3 a_k |\xi(k)\ran\lan\xi(k)|x_-^\perp\ran\ot |y_-^\perp\ran\ot |z_-^\perp\ran=
\left(a_0\lan x_+ |x_-^\perp\ran  \lan y_+|y_-^\perp\ran \lan z_+|z_-^\perp\ran\right)|\xi(0)\ran,
$$
and so we have $a_0=0$.
We also have $a_1=a_2=a_3=0$ in the same say, by applying
$$
|x_-^\perp\ran\ot |y_+^\perp\ran \ot |z_+^\perp\ran,\qquad
|x_+^\perp\ran\ot |y_-^\perp\ran \ot |z_+^\perp\ran,\qquad
|x_+^\perp\ran\ot |y_+^\perp\ran \ot |z_-^\perp\ran,
$$
respectively.

We consider the span $E$ of $\{|\xi(k)\rangle:k=0,1,2,3\}$, and
suppose that a product vector $|\zeta\rangle=|\zeta_1\rangle\otimes
|\zeta_2\rangle\otimes |\zeta_3\rangle$ is in the orthogonal
complement $E^\perp$. We first note that $|\zeta_1\rangle$ must be
orthogonal to $|x_+\rangle$ or $|x_-\rangle$, that is,
$|\zeta_1\rangle \para |x_+^\perp\rangle$ or
$|\zeta_1\rangle \para |x_-^\perp\rangle$ in notation for
parallel. If $|\zeta_1\rangle \para |x_+^\perp\rangle$, then we
have two possibilities:
\begin{itemize}
\item
both $|\zeta_2\rangle\para |y_+^\perp\rangle$ and $|\zeta_3\rangle\para |z_+^\perp\rangle$
hold,
\item
both $|\zeta_2\rangle\para |y_-^\perp\rangle$ and
$|\eta_3\rangle\para |z_-^\perp\rangle$ hold.
\end{itemize}
Therefore,
$|\zeta\rangle$ must be parallel to $|x_+^\perp\rangle\otimes
|y_+^\perp\rangle \otimes |z_+^\perp\rangle$ or
$|x_+^\perp\rangle\otimes |y_-^\perp\rangle \otimes
|z_-^\perp\rangle$. Now, we write $|\eta\rangle
=|x^\perp\rangle \otimes |y^\perp\rangle \otimes |z^\perp\rangle$. We have
shown that if $|\zeta_1\rangle\para |x_+^\perp\rangle$ then
$|\zeta\rangle\para |\eta(0)\rangle$ or $|\zeta\rangle\para
|\eta(1)\rangle$. In the same way, if $|\zeta_1\rangle\para
|x_-^\perp\rangle$ then we have $|\zeta\rangle\para
|\eta(2)\rangle$ or $|\zeta\rangle\para |\eta(3)\rangle$.
Therefore, we see that $E^\perp$ has only four product vectors
$|\eta(k)\rangle$ with $k=0,1,2,3$. Again, the orthogonal complement
$E=(E^\perp)^\perp$ of $E^\perp$ has only four product vectors $|\xi(k)\rangle$
with $k=0,1,2,3$. This completes the proof.
\end{proof}

Uniqueness of decomposition into pure product states has been considered by several authors.
For examples, see \cite{{alfsen},{ha-kye-sep-face},{kirk}}
for bi-partite cases, and \cite{ha-kye-multi-unique} for multi-qubit cases.
Especially, it was shown in Theorem 4.1 of \cite{ha-kye-multi-unique}
that generic choice of four product vectors in the three-qubit system
gives rise to a separable state with unique decomposition. We could not use this result, because four vectors in
the range of $\varrho$ are not in general position in the above discussion.
Note that separable states with rank four have been studied extensively in \cite{chen_dj_rank4}.

\begin{theorem}\label{rank-four}
Suppose that $\varrho=X(a,b,c)$ is a non-diagonal {\sf X}-state of rank four. Then the following are equivalent:
\begin{enumerate}
\item[(i)]
$\varrho$ is separable;
\item[(ii)]
$\varrho$ satisfies the relations
\begin{equation}\label{rank-4}
a_1a_4=a_2a_3,\qquad \sqrt{a_ib_i}=|c_j|\ (i,j=1,2,3,4),\qquad c_1 c_4  = c_2c_3.
\end{equation}
\item[(iii)]
there exists a product vector $|\xi\ran = |x\ran\ot |y\ran \ot |z\ran$
with nonzero entries such that $\varrho$ is the {\sf X}-part of $|\xi\ran\lan \xi|$.
\end{enumerate}
If $|\eta \rangle$ is another product vector satisfying {\rm (iii)}, then we have $|\eta\ran=|\xi(k)\ran$
up to scalar multiplication for
some $k=0,1,2,3$.
\end{theorem}

\begin{proof}
Suppose that $\varrho$ is separable, and so $\varrho$ satisfies the inequality (\ref{ppt-diag}).
We first note that $\varrho$ is of rank four if and only if $\sqrt{a_ib_i}=|c_i|$ for each $i=1,2,3,4$.
This shows that the inequality $\sqrt{a_ib_i}\ge |c_j|$ given by the PPT condition actually becomes the identity
$\sqrt{a_ib_i}=|c_j|$ for $i,j=1,2,3,4$. Write this number by $R$. By (\ref{ppt-diag}), we have
$$
R^2\le \sqrt[4]{a_1b_2b_3a_4}\, \sqrt[4]{b_1a_2a_3b_4}=\prod_{i=1}^4\sqrt[4]{a_ib_i}=R^2,
$$
which implies the identity $\sqrt[4]{a_1b_2b_3a_4}=R=\sqrt[4]{b_1a_2a_3b_4}$. Therefore, we have
$$
a_1a_4={R^4 \over b_1 b_4}=a_2a_3.
$$
Especially, we see that the relation
$\Delta_\varrho=R_\varrho=r_\varrho$ holds, and so the last condition of (\ref{rank-4}) must
hold by Corollary \ref{phase}. This proves the implication (i) $\Longrightarrow$ (ii).

For the direction (ii) $\Longrightarrow$ (iii), we may assume that $|c_i|=1$
and write $c_i=e^{{\rm i}\theta_i}$ for each $i=1,2,3,4$. We put
\begin{equation}\label{anglesolution}
\alpha={-\theta_2-\theta_3 \over 2},\quad \beta={-\theta_2+\theta_4 \over 2},\quad \gamma={-\theta_3+\theta_4 \over 2},
\end{equation}
and
$$
|\xi \rangle = {1 \over a_1}(\sqrt{a_1}, \sqrt{b_4} e^{{\rm i} \alpha})^\ttt
\otimes (\sqrt{a_1},\sqrt{a_3} e^{{\rm i} \beta})^\ttt \otimes (\sqrt{a_1},\sqrt{a_2} e^{{\rm i} \gamma})^\ttt.
$$
It is directly checked that the {\sf X}-part of $|\xi\rangle \langle \xi|$ is $\varrho$.
The direction (iii) $\Longrightarrow$ (i) follows by Proposition \ref{xpart-1}.

Suppose that $|\eta \rangle$ is another product vector satisfying {\rm (iii)}. We may write
$$
|\eta\ran=(x_0,x_1 e^{{\rm i}\alpha})^\ttt \ot (y_0,y_1 e^{{\rm i}\beta})^\ttt\ot (z_0,z_1 e^{{\rm i}\gamma})^\ttt,
$$
with positive numbers $x_i,y_i$ and $z_i$.
Let $|x\ran=(x_0,x_1)^\ttt, |y\ran=(y_0,y_1)^\ttt$ and $|z\ran=(z_0,z_1)^\ttt$ and
$|\zeta \rangle = |x\ran\ot |y\ran \ot |z\ran$.
Since $|\zeta\rangle\langle \zeta|$ shares the diagonal with $|\eta\rangle\langle\eta|$, we have
\begin{equation}\label{sch-1}
|\zeta\ran=(\sqrt a_1, \sqrt a_2,\sqrt a_3,\sqrt a_4, \sqrt b_4,\sqrt b_3,\sqrt b_2\sqrt b_1)^\ttt.
\end{equation}
The relation (\ref{rank-4}) shows that the following two matrices
$$
\left( \begin{matrix} \sqrt a_1 &\sqrt a_2 \\ \sqrt a_3 &\sqrt a_4\end{matrix} \right),
\qquad
\left(\begin{matrix}
\sqrt a_1 &\sqrt a_2 &\sqrt a_3 &\sqrt a_4\\ \sqrt b_4 &\sqrt b_3 &\sqrt b_2 &\sqrt b_1\end{matrix}\right)
$$
are of column rank one. From their row rank, we have
$$
|x\ran=(\sqrt a_1,\sqrt b_4)^\ttt,\qquad
|y\ran = (\sqrt a_1,\sqrt a_3)^\ttt,\qquad
|z\ran=(\sqrt a_1,\sqrt a_2)^\ttt,
$$
up to scalar multiplication. From the anti-diagonals of $|\eta \rangle \langle \eta |$,
we get the equations
$$
\alpha+\beta-\gamma=-\theta_2,\quad \alpha-\beta+\gamma=-\theta_3,\qquad -\alpha+\beta+\gamma=\theta_4.
$$
Their solutions are given by (\ref{anglesolution}). If we replace one of $\theta_i$ as $\theta_i+2\pi$,
then two of (\ref{anglesolution}) are changed up to $\pm\pi$. Hence, $|\eta\ran$ must be
one of $|\xi(k)\ran$, up to scalar multiplication.
\end{proof}

We note that the relation (\ref{rank-4}) in fact tells us that the right side of vector (\ref{sch-1})
has the Schmidt rank $(1,1,1)$ in the sense of \cite{han_kye_tri}.
The following sufficient condition for separability will be useful in the next section.

\begin{corollary}\label{suff-4}
A three-qubit {\sf X}-state $\varrho=X(a,b,c)$ is separable whenever
$$
\Delta_\varrho\ge R_\varrho=r_\varrho,\qquad c_1c_4=c_2c_3.
$$
\end{corollary}

\begin{proof}
If $\varrho$ is diagonal, then there is nothing to prove.
When $\varrho$ is not diagonal, we have $R_\varrho=r_\varrho>0$, and may assume that
$R_\varrho=r_\varrho=1$ by multiplying a scalar.
We consider the following three intervals
$$
I_1=[1/a_1,b_1],\qquad I_2=[1/b_2,a_2],\qquad I_3=[1/b_3,a_3],
$$
which are nonempty by the assumption $a_ib_i\ge 1$ for $i=1,2,3$. Therefore, we see that
every real number in the interval $[1/a_1b_2b_3, b_1a_2a_3]$ can be written by the product
$x_1x_2x_3$ of real numbers $x_i\in I_i$ for $i=1,2,3$. Now, we see that the intersection
$[1/b_4,a_4]\cap [1/a_1b_2b_3, b_1a_2a_3]$ of the two intervals is nonempty by the assumption
$a_1b_2b_3a_4\ge 1$ and $b_1a_2a_3b_4\ge 1$. Hence, we can take $x_i$ satisfying
$$
\frac 1{a_1}\le x_1\le b_1,\quad
\frac 1{b_2}\le x_2\le a_2,\qquad
\frac 1{b_3}\le x_3\le a_3,\qquad
\frac 1{b_4}\le x_1x_2x_3\le a_4.
$$
Now, we define $a_i^\prime$ and $b^\prime_i$ by
$$
a_1^\prime=\frac 1{x_1},\quad a^\prime_2=x_2,\quad a^\prime_3=x_3,\quad a^\prime_4=x_1x_2x_3,
\quad b^\prime_i=1/a^\prime_i\ (i=1,2,3,4).
$$
Then we have $a^\prime_i\le a_i$ and $b^\prime_i\le b_i$ for $i=1,2,3,4$, and the {\sf X}-state
$\varrho^\prime=X(a^\prime, b^\prime,c)$ is a separable state of rank four by Theorem \ref{rank-four}. Therefore, we see that
$\varrho$ is the sum of a separable state $\varrho^\prime$ and a diagonal state.
\end{proof}

We denote by $\mathbb S$ the convex set of all three-qubit separable states.
A separable state $\varrho$ determines a unique face $F_\varrho$ of $\mathbb S$
such that $\varrho$ is an interior point of $F_\varrho$. This is the smallest face containing $\varrho$.
A nonempty convex subset $F$ of a convex set $C$ is said to be a face of $C$
if $x,y\in F$ whenever a nontrivial convex combination of $x,y\in C$ belongs to $F$.
A separable state $\varrho$
has a unique decomposition into pure product states if and only if $F_\varrho$ is a simplex.
In this case, the set of extreme points of $F_\varrho$ consists of pure product states in the decomposition.
Theorem \ref{rank-four} tells us that
if $\varrho$ is a non-diagonal {\sf X}-state of rank four, then $F_\varrho$ is a simplicial face,
that is, a face which is a simplex,
whose extreme points arise from four product vectors in the set
$$
\Pi_\xi=\{|\xi(k)\ran: k=0,1,2,3\}
$$
listed in (\ref{decom-rank1}) for a product vector $|\xi\ran$.
We note that $\Pi_{\xi(k)}=\Pi_\xi$ for $k=0,1,2,3$.

\begin{proposition}
If $\varrho$ is a non-diagonal three-qubit separable {\sf X}-state of rank four, then
the face $F_\varrho$ coincides with the set $\{\omega \in \mathbb S : \omega_X = \varrho \}$.
\end{proposition}

\begin{proof}
Take a product vector $|\xi\ran$ so that $\varrho$ is the {\sf X}-part of $|\xi\ran\lan\xi|$.
Then $F_\varrho$ consists of convex combinations of $|\xi(k)\ran\lan\xi(k)|$ with $k=0,1,2,3$.
Since the {\sf X}-part of each $|\xi(k)\ran\lan\xi(k)|$ is $\varrho$,
we see that $F_\varrho$ is a subset of $\{\omega \in \mathbb S : \omega_X = \varrho \}$.
Suppose that the {\sf X}-part of $\omega = \sum_i \lambda_i |\eta_i\ran\lan\eta_i|$ coincides with $\varrho$.
By (\ref{average}), we have
$$
\varrho = \omega_X = \sum_i {\lambda_i \over 4}
(|\eta_i(0)\ran\lan\eta_i(0)|
+|\eta_i(1)\ran\lan\eta_i(1)|
+|\eta_i(2)\ran\lan\eta_i(2)|
+|\eta_i(3)\ran\lan\eta_i(3)|).
$$
By the uniqueness part of Theorem \ref{rank-four}, each $|\eta_i\ran$ is one of $|\xi(k)\ran$,
and so we conclude that $\omega$ belongs to $F_\varrho$.
\end{proof}

We have established one-to-one correspondence between the following objects:
\begin{itemize}
\item
non-diagonal three-qubit separable {\sf X}-states $\varrho$ of rank four,
\item
the set $\Pi_\xi$ of four product vectors with nonzero entries,
\item
the $3$-dimensional simplicial faces $F_\varrho$ determined by $\Pi_\xi$.
\end{itemize}
The {\sf X}-state $\varrho$ is located at the center of the simplex $F_\varrho$.
If $\varrho_1$ and $\varrho_2$ are distinct non-diagonal separable {\sf X}-states
of rank four, then two faces $F_{\varrho_1}$ and $F_{\varrho_2}$ have no intersection by the above proposition.
If we denote  by $\mathbb S_{\text{\sf X}}$ the convex set of three-qubit separable {\sf X}-states, then $\mathbb
S_{\text{\sf X}}\cap F_\varrho$ consists of a single point
$\varrho$. This shows that $\varrho$ is an extreme point of the
convex set $\mathbb S_{\text{\sf X}}$. On the other hand, Proposition \ref{extreme}
tells us that an extreme point of $\mathbb S_{\text{\sf X}}$ is a non-diagonal state of rank four
or a diagonal state.

\begin{theorem}
A three-qubit separable {\sf X}-state is an extreme point of the convex set $\mathbb S_{\text{\sf X}}$ if and only if
it is non-diagonal of rank four, or diagonal of rank one.
\end{theorem}

%%%%%%%%%%%%%%%%%%%%%%%%%%%%%%%%%%%%%%%%%%%%%%%%%%%%%%%%%%%%%%%%%%%%%%%%%%%%%%%%%%%%%%%%%%%%%%%%%%%%%%%%%%%%%%%%%%
%%%%%%%%%%%%%%%%%%%%%%%%%%%%%%%%%%%%%%%%%%%%%%%%%%%%%%%%%%%%%%%%%%%%%%%%%%%%%%%%%%%%%%%%%%%%%%%%%%%%%%%%%%%%%%%%%%
%%%%%%%%%%%%%%%%%%%%%%%%%%%%%%%%%%%%%%%%%%%%%%%%%%%%%%%%%%%%%%%%%%%%%%%%%%%%%%%%%%%%%%%%%%%%%%%%%%%%%%%%%%%%%%%%%%
%%%%%%%%%%%%%%%%%%%%%%%%%%%%%%%%%%%%%%%%%%%%%%%%%%%%%%%%%%%%%%%%%%%%%%%%%%%%%%%%%%%%%%%%%%%%%%%%%%%%%%%%%%%%%%%%%%
%%%%%%%%%%%%%%%%%%%%%%%%%%%%%%%%%%%%%%%%%%%%%%%%%%%%%%%%%%%%%%%%%%%%%%%%%%%%%%%%%%%%%%%%%%%%%%%%%%%%%%%%%%%%%%%%%%
\section{Sufficient criteria for separability}

In this section, we consider the case of $R_\varrho=r_\varrho$, and show that the
inequality {\rm (\ref{Phi})} is sufficient for separability of
an {\sf X}-state $\varrho=X(a,b,c)$. Actually, we will express $\varrho$ as the mixture of
two {\sf X}-states satisfying the conditions in Corollary \ref{suff-4}.

We assume that
$|c_i|=1$ and $0\le\phi_\varrho<2\pi$ without loss of generality. In this case, we have
$A_\varrho=\sqrt{1+\sin \phi_\varrho /2}$ by Lemma \ref{T_theta}. For
notational convenience, write $\phi=\phi_\varrho/2$ and
$r= A_\varrho=\sqrt{1+\sin \phi}=\sin \phi/2 +\cos \phi/2$. Then we have
$$
\begin{aligned}
\tan {\textstyle{\phi \over 2}}~ \left( 1-re^{{\rm i}(\phi/2-\pi/2)} \right)
= &\tan {\textstyle{\phi \over 2}}~ \left( 1+ {\rm i} re^{{\rm i}(\phi/2)} \right) \\
= &\tan {\textstyle{\phi \over 2}}~ \left( 1- r\sin {\textstyle{\phi
\over 2}}
  + {\rm i} r \cos {\textstyle{\phi \over 2}} \right) \\
= &\tan {\textstyle{\phi \over 2}}~ \left( \cos^2 {{\textstyle{\phi
\over 2}}}
   - \sin {\textstyle{\phi \over 2}} \cos {\textstyle{\phi \over 2}} \right)
     + {\rm i} r \sin {\textstyle{\phi \over 2}} \\
= &r \cos {\textstyle{\phi \over 2}} -1 + {\rm i} r \sin {{\textstyle{\phi \over 2}}}\\
= &re^{{\rm i}(\phi/2)}-1,
\end{aligned}
$$
which tells us that the three points $re^{{\rm i}(\phi/2)}, 1,
re^{{\rm i}(\phi/2-\pi/2)}$ are co-linear on the complex plane.
See FIGURE 3.

\setlength{\unitlength}{.8 mm}
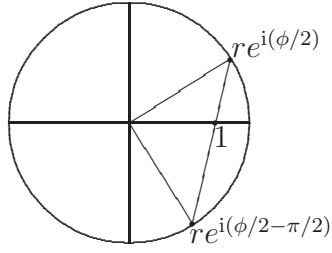
\begin{figure}
\setlength{\unitlength}{.8 mm}
\begin{center}
\begin{picture}(50,50)
  \qbezier(25.000,0.000)(33.284,0.000)
          (39.142,5.858)
  \qbezier(39.142,5.858)(45.000,11.716)
          (45.000,20.000)
  \qbezier(45.000,20.000)(45.000,28.284)
          (39.142,34.142)
  \qbezier(39.142,34.142)(33.284,40.000)
          (25.000,40.000)
  \qbezier(25.000,40.000)(16.716,40.000)
          (10.858,34.142)
  \qbezier(10.858,34.142)( 5.000,28.284)
          ( 5.000,20.000)
  \qbezier( 5.000,20.000)( 5.000,11.716)
          (10.858,5.858)
  \qbezier(10.858,5.858)(16.716,0.000)
          (25.000,0.000)
\drawline(5,20)(45,20)
\drawline(25,0)(25,40)
\drawline(25,20)(35.5,3.2)
\drawline(25,20)(41.8,30.5)
\drawline(35.5,3.2)(41.8,30.5)
\put(35.5,3.2){\circle*{1}}
\put(41.8,30.5){\circle*{1}}
\put(39.3,20){\circle*{1}}
\put(39,16){$1$}
\put(42,31){$re^{{\rm i}(\phi/2)}$}
\put(35,0){$re^{{\rm i}(\phi/2-\pi/2)}$}
\end{picture}
\end{center}
\caption{Three complex numbers $re^{{\rm i}(\phi/2)}$, $1$ and  $re^{{\rm i}(\phi/2-\pi/2)}$
are co-linear on the complex plane. Note that $0\le\phi/2<\pi/2$.}
\end{figure}

Therefore, we can take nonnegative numbers $p$ and $q$ such that
$$
p+q=1,\qquad 1=p re^{{\rm i}(\phi/2)} + q re^{{\rm i}(\phi/2-\pi/2)}.
$$
If we put $u=re^{{\rm i}(\phi/2)}$ and $v=re^{{\rm i}(\phi/2-\pi/2)}$, then we have $pu+qv=1$.
Put $\varrho_1=X(a,b, c^{\prime})$ and $\varrho_2=X(a,b, c^{\prime\prime})$ with
$$
c^\prime=(c_1 \bar u, c_2 u, c_3 u, c_4 \bar u),\qquad
c^{\prime\prime}=(c_1\bar v, c_2 v, c_3 v, c_4 \bar v).
$$
Now, we suppose that $\varrho$ satisfies the inequality (\ref{theta-delta}). Then we have
$$
\varrho= p\varrho_1+q\varrho_2,\qquad
\Delta_{\varrho_i}=\Delta_{\varrho}\ge  A_\varrho=r=  R_{\varrho_i},
$$
for $i=1,2$. Furthermore, we have
$$
\begin{aligned}
(c_1 \bar u)(c_4 \bar u)
&= |c_1||c_4|r^2 e^{{\rm i} (\theta_1+\theta_4 - \phi)}\\
&= |c_2||c_3|r^2 e^{{\rm i} (\theta_2+\theta_3 + \phi)} = (c_2 u)(c_3 u),
\end{aligned}
$$
and $(c_1 \bar v)(c_4 \bar v)=(c_2 v)(c_3 v)$ similarly.
Therefore, we see that $\varrho_1$ and $\varrho_2$ are separable by Corollary
\ref{suff-4}, and so {\rm (\ref{theta-delta})} implies the
separability of $\varrho$ when $R_\varrho=r_\varrho$.
We know that the separability of $\varrho$ implies {\rm (\ref{theta-delta})},
and two criteria {\rm (\ref{theta-delta})} and  (\ref{Phi}) are
equivalent when $R_\varrho=r_\varrho$ by Lemma \ref{T_theta}. Therefore, we have the following:

\begin{theorem}\label{main-equiv}
Let $\varrho=X(a,b,c)$ be a three-qubit {\sf X}-state with a common anti-diagonal magnitude.
Then $\varrho$ is separable if and only if $\varrho$ satisfies the inequality {\rm (\ref{theta-delta})} if and only if
$\varrho$ satisfies the inequality {\rm (\ref{Phi})}.
\end{theorem}

We return to the four dimensional convex body
$\mathbb S$ consisting of $(r_1,r_2,r_3,r_4)$ so that the {\sf
X}-state $X(a,b,c)$ is separable. See FIGURE 2 again. Theorem \ref{main-equiv}
tells us that the \lq corner point\rq\ of the union of strips belongs to the
region $\mathbb S$, and so, the bound (\ref{Phi}) for the minimum of
anti-diagonal magnitudes is optimal. Furthermore, the \lq corner point\rq\
represents a boundary separable state with full ranks whenever
$\phi_\varrho\neq 0$. Construction of such a state has been asked in
\cite{chen_full_sep} and answered in \cite{kye_3qb_EW}. We add here
more examples.

In general cases with arbitrary anti-diagonal magnitudes, we have the following sufficient condition
for separability, which tells us that the the region $\mathbb S$ lies between two cubes
determined by (\ref{ppt-diag}) and (\ref{suff-Phi}) below:

\begin{proposition}
Let $\varrho=X(a,b,c)$ be a three-qubit {\sf X}-state. Then $\varrho$ is separable whenever
it satisfies the following inequality:
\begin{equation}\label{suff-Phi}
\Delta_\varrho \ge R_\varrho \sqrt{1+\max\{|\sin\phi_{\varrho}/2|, |\cos\phi_{\varrho}/2|\}}
\end{equation}
\end{proposition}

\begin{proof}
For a string $\epsilon=\epsilon_1\epsilon_2\epsilon_3\epsilon_4$ of $\pm 1$, we
consider the {\sf X}-state $\varrho^\epsilon=X(a,b,
R_\varrho\Phi^\epsilon_\varrho)$ with the phase part
$\Phi^\epsilon_\varrho= (\epsilon_1e^{{\rm i}\theta_1},
\epsilon_2e^{{\rm i}\theta_2}, \epsilon_3e^{{\rm i}\theta_3},
\epsilon_4e^{{\rm i}\theta_4})$. Since
$$
|\sin\phi_{\varrho^\epsilon}/2|\le \max\{|\sin\phi_\varrho/2|,|\cos\phi_\varrho/2|\}
$$
in general, the
inequality (\ref{suff-Phi}) tells us that the criterion (\ref{Phi})
holds for the state $\varrho^\epsilon$. Because each
$\varrho^\epsilon$ shares a common anti-diagonal magnitude
$R_\varrho$, we see that $\varrho^\epsilon$ is separable by Theorem \ref{main-equiv} for each
string $\epsilon$. This shows that the state $\varrho$ is also
separable since $\varrho$ is a convex combination of
$\varrho^\epsilon$'s.
\end{proof}

{\sl Example 1}:
Consider the {\sf X}-shaped matrix
$$
\varrho_{r,\theta}=X({\bf 1}, {\bf 1}, (r,r,re^{{\rm i}\theta},r)),
$$
where ${\bf 1}=(1,1,1,1)$.
Then $\varrho_{r,\theta}$ is a state if and only if it is a PPT state if and only if
$r\le 1$. On the other hand, our main criterion (\ref{Phi}) tells us that $\varrho_{r,\theta}$ is
separable if and only if
$$
1\ge r\sqrt{1+|\sin(\theta/2)|}.
$$
This example shows clearly the role of phases for the criteria of separability.
See FIGURE 4. When $\theta=0$ and $c=(r,r,r,r)$, we see that $\varrho_{r,0}$ is separable if and only if $r\le 1$.
On the other hand, in the case of $\theta=\pi$ and $c=(r,r,-r,r)$, it is known \cite{{guhne_pla_2011},{han_kye_GHZ},{kay_2011}}
that $\varrho_{r,\pi}$ is separable if and only if $r\le \frac 1{\sqrt 2}$.
Our result interpolates these two boundary separable states $\varrho_{1,0}$ and $\varrho_{1/\sqrt 2,\, \pi}$
to get a one parameter family $\varrho_{r,\theta}$ of boundary separable states, with the curve $r=1/\sqrt{1+|\sin(\theta/2)|}$.
$\square$\par\medskip

\setlength{\unitlength}{.6 mm}
\begin{figure}
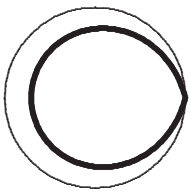

\input 3-qb-phase_pic_3.tex
\end{picture}
\end{center}
\caption{The thick curve is $1=r\sqrt{1+|\sin(\theta/2)|}$ on the
complex plane with the polar coordinate, which represents the
boundary of the region for separability of the {\sf X}-state $\varrho_{r,\theta}=X({\bf
1},{\bf 1}, (r,r,re^{{\rm i}\theta},r))$. The circle represents the
region of PPT property.}
\end{figure}

{\sl Example 2}:
We examine various criteria for the {\sf X}-shaped matrix
$$
\varrho_{p,q} =X({\bf 1},{\bf 1},(p,p,q,-q))
$$
with real numbers $p$ and $q$. We note that
$\varrho_{p,q}$ is a state if and only if it is of PPT if and only if
$1\ge \max\{|p|,|q|\}$. This a GHZ diagonal state. The inequality
(\ref{theta-delta}) is just $1\ge (|p|+|q|)/\sqrt 2$, which does not
detect entanglement when $|p|/|q|$ is big or small enough. One can
also check that $\varrho_{p,q}$ is separable if and only if $1\ge
\sqrt{p^2+q^2}$ by the result in \cite{han_kye_GHZ}. See FIGURE 5.
$\square$\par\medskip

\begin{figure}\label{figqqqqq}
\setlength{\unitlength}{.6 mm}
\begin{center}
\begin{picture}(50,80)
  \qbezier(25.000,20.000)(33.284,20.000)
          (39.142,25.858)
  \qbezier(39.142,25.858)(45.000,31.716)
          (45.000,40.000)
  \qbezier(45.000,40.000)(45.000,48.284)
          (39.142,54.142)
  \qbezier(39.142,54.142)(33.284,60.000)
          (25.000,60.000)
  \qbezier(25.000,60.000)(16.716,60.000)
          (10.858,54.142)
  \qbezier(10.858,54.142)( 5.000,48.284)
          ( 5.000,40.000)
  \qbezier( 5.000,40.000)( 5.000,31.716)
          (10.858,25.858)
  \qbezier(10.858,25.858)(16.716,20.000)
          (25.000,20.000)
\drawline(39.142,25.858)(59.142,25.858)
\drawline(39.142,25.858)(39.142,5.858)
\drawline(39.142,54.142)(59.142,54.142)
\drawline(39.142,54.142)(39.142,74.142)
\drawline(10.858,54.142)(-10.142,54.142)
\drawline(10.858,54.142)(10.858,74.142)
\drawline(10.858,25.858)(-10.142,25.858)
\drawline(10.858,25.858)(10.858,5.858)
\drawline(25,11.7)(53.3,40)
\drawline(25,68.3)(53.3,40)
\drawline(25,68.3)(-3.3,40)
\drawline(25,11.7)(-3.3,40)
\end{picture}
\end{center}
\caption{The circle centered at the origin on the $pq$-plane
represents the region for separability of $\varrho_{p,q}=X({\bf 1},{\bf
1},(p,p,q,-q))$. The diamond and the union of horizontal and
vertical strips represent the regions satisfying the inequalities
(\ref{theta-delta}) and (\ref{Phi}), respectively. Two cubes by the
conditions (\ref{ppt-diag}) and (\ref{suff-Phi}) are squares (which
are not shown in the picture) circumscribing and inscribing the
circle, respectively.}
\end{figure}
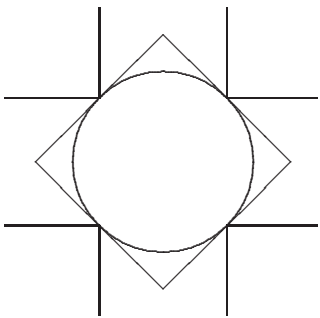

We compare two separability criteria (\ref{Phi}) and
(\ref{theta-delta}). The criterion (\ref{Phi}) shows the role of
phases more directly, but it is weaker criterion than
(\ref{theta-delta}) by the inequality (\ref{yyyyy}). These two
criteria are equivalent to each other when the anti-diagonal entries
share a common magnitude, by Lemma \ref{T_theta}. We note that the inequality $r_\varrho\le R_\varrho\le
\Delta_\varrho$ holds for general separable states. Therefore, the smaller is the ratio
$R_\varrho/ r_\varrho$, the sharper is the criterion (\ref{Phi}). In
fact, we have shown in Theorem \ref{main-equiv} that each of two criteria also gives rise to a
sufficient condition for separability of {\sf X}-states when
$R_\varrho= r_\varrho$. On the other hand, they are of little use to
detect entanglement when the ratio $R_\varrho/r_\varrho$ is big, as
we have seen in the example $\varrho_{p,q}=X({\bf 1},{\bf 1},(p,p,q,-q))$.

We note that these two criteria (\ref{Phi}) and (\ref{theta-delta}) depend on the criterion
({\ref{guhne}) which is not so easy to apply directly. In the case
when all the $z_i$'s are real, the maximum part in the criterion ({\ref{guhne}) can
be evaluated in terms of $z_i$'s. This was useful in
\cite{han_kye_GHZ} to characterize separability of GHZ diagonal
states, where all the anti-diagonal entries are real numbers.

%%%%%%%%%%%%%%%%%%%%%%%%%%%%%%%%%%%%%%%%%%%%%%%%%%%%%%%%%%%%%%%%%%%%%%%%%%%%%%%%%%%%%%%%%%%%%%%%%%%%%%%%%%%%%%%%
%%%%%%%%%%%%%%%%%%%%%%%%%%%%%%%%%%%%%%%%%%%%%%%%%%%%%%%%%%%%%%%%%%%%%%%%%%%%%%%%%%%%%%%%%%%%%%%%%%%%%%%%%%%%%%%%
%%%%%%%%%%%%%%%%%%%%%%%%%%%%%%%%%%%%%%%%%%%%%%%%%%%%%%%%%%%%%%%%%%%%%%%%%%%%%%%%%%%%%%%%%%%%%%%%%%%%%%%%%%%%%%%%
%%%%%%%%%%%%%%%%%%%%%%%%%%%%%%%%%%%%%%%%%%%%%%%%%%%%%%%%%%%%%%%%%%%%%%%%%%%%%%%%%%%%%%%%%%%%%%%%%%%%%%%%%%%%%%%%
%%%%%%%%%%%%%%%%%%%%%%%%%%%%%%%%%%%%%%%%%%%%%%%%%%%%%%%%%%%%%%%%%%%%%%%%%%%%%%%%%%%%%%%%%%%%%%%%%%%%%%%%%%%%%%%%
%%%%%%%%%%%%%%%%%%%%%%%%%%%%%%%%%%%%%%%%%%%%%%%%%%%%%%%%%%%%%%%%%%%%%%%%%%%%%%%%%%%%%%%%%%%%%%%%%%%%%%%%%%%%%%%%
%%%%%%%%%%%%%%%%%%%%%%%%%%%%%%%%%%%%%%%%%%%%%%%%%%%%%%%%%%%%%%%%%%%%%%%%%%%%%%%%%%%%%%%%%%%%%%%%%%%%%%%%%%%%%%%%
%%%%%%%%%%%%%%%%%%%%%%%%%%%%%%%%%%%%%%%%%%%%%%%%%%%%%%%%%%%%%%%%%%%%%%%%%%%%%%%%%%%%%%%%%%%%%%%%%%%%%%%%%%%%%%%%
%%%%%%%%%%%%%%%%%%%%%%%%%%%%%%%%%%%%%%%%%%%%%%%%%%%%%%%%%%%%%%%%%%%%%%%%%%%%%%%%%%%%%%%%%%%%%%%%%%%%%%%%%%%%%%%%
\section{Rank five and six cases}

We have seen that every non-diagonal separable {\sf X}-state with rank four should satisfy the phase identity:
$\theta_1+\theta_4=\theta_2+\theta_3$.
In this section, we show that this is also the case whenever the rank is five or six. We stress here that
a separable {\sf X}-state of rank six may not satisfy the identity $c_1c_4=c_2c_3$
even though it obeys the phase identity.

\begin{theorem}\label{main-rank-6}
Let $\varrho=X(a,b,c)$ be a three-qubit non-diagonal {\sf X}-state of $\rk\varrho\le 6$.
Then $\varrho$ is separable if and only if the following hold:
\begin{enumerate}
\item[(i)]
the relation {\rm (\ref{ppt-diag})} holds;
\item[(ii)]
there exists a partition $\{i_1,i_2\}\sqcup\{i_3,i_4\}=\{1,2,3,4\}$ such that
$$
\sqrt{a_i b_i} = R_\varrho = |c_i|  \ (i=i_1, i_2),\qquad \sqrt{a_j b_j} \ge R_\varrho \ge r_\varrho=|c_j|,  \ (j=i_3,i_4);
$$
\item[(iii)]
if $r_\varrho>0$, then the phase identity holds.
\end{enumerate}
\end{theorem}

\begin{proof}
By the rank condition and non-diagonality, we see that there exists $\{i_1,i_2\}$ such that
$\sqrt{a_ib_i}=|c_i|$ for $i=i_1,i_2$.
Suppose that $\varrho$ is separable. We first note that $\varrho$ satisfies the condition (\ref{ppt-diag}),
from which we see that the first condition of (ii) holds.
Now, we write
$\varrho = \sum_k \lambda_k \omega_k$
with pure product states $\omega_k$'s, where $\lambda_k>0$ and $\sum_k \lambda_k=1$.
Suppose that the {\sf X}-part of $\omega_k$ is given by $X(d^k,e^k,f^k)$.
Then we have the inequalities:
\begin{equation}\label{jgilugjkhkj}
\begin{aligned}
|c_i| = |\sum_k \lambda_k f^k_i| & \le \sum_k \lambda_k |f^k_i| \\
& \le \sum_k \lambda_k \sqrt{d^k_i e^k_i} \\
& \le (\sum_k \lambda_k d^k_i)^{1 \slash 2} (\sum_k \lambda_k e^k_i)^{1 \slash 2} = \sqrt{a_i b_i},
\end{aligned}
\end{equation}
for $i=1,2,3,4$.
Therefore, all the above inequalities must be identities for $i=i_1,i_2$.
By the first identity, we have $\arg c_i=\arg f^k_i$ for $i=i_1,i_2$ and all $k$ with $f_i^k \ne 0$.

We first consider the case $\{i_1,i_2\}=\{1,2\}$.
By Theorem \ref{rank-four}, we have $|f_i^k|=|f_j^k|$ and $f^k_1 f^k_4 = f^k_2 f^k_3$,
which implies $e^{{\rm i}\theta_1} f^k_4 = e^{{\rm i}\theta_2} f^k_3$.
Therefore, we have
$$
c_1 c_4 = |c_1|~ \sum_k \lambda_k e^{{\rm i} \theta_1} f^k_4
= |c_2|~ \sum_k \lambda_k e^{{\rm i} \theta_2} f^k_3 = c_2c_3.
$$
Especially, we have $|c_3|=|c_4|$ and the phase identity. This completes the proof of \lq only if\rq\ part,
for the case of $\{i_1,i_2\}=\{1,2\}$. For the converse, we will express
$c_i$ ($i=3,4$) by the average of two complex numbers with absolute value $R_\varrho$. To do this, define the phase $\phi$ by the relation
$\cos\phi=r_\varrho \slash R_\varrho$, and put
$\alpha_i = \phi + \theta_i$ and $\beta_i = -\phi +\theta_i$ for $i=3,4$. See FIGURE 6.
In the case $r_\varrho=0$, take arbitrary $\theta_3$ and $\theta_4$ satisfying the phase identity.
Then we have
${1 \over 2}(R_\varrho e^{{\rm i}\alpha_i}+R_\varrho e^{{\rm i}\beta_i})=c_i$ for $i=3,4$,
and
$\theta_1+\alpha_4=\theta_2+\alpha_3$ together with
$\theta_1+\beta_4=\theta_2+\beta_3$.
By condition (\ref{ppt-diag}) on the diagonal part,
we can take $\lambda$ so that
$$
{a_1 \over a_3}, {a_2 \over a_4} \le \lambda \le {b_4 \over b_2}, {b_3 \over b_1}.
$$
Now, we put $a^\prime=(a_1, a_2, a_1 / \lambda, a_2 / \lambda)$ and $b^\prime=(b_1, b_2, \lambda b_1, \lambda b_2)$.
We also put $c^\prime=(c_1, c_2, R_\varrho e^{{\rm i}\alpha_3}, R_\varrho e^{{\rm i}\alpha_4})$
and $c^{\prime\prime}=(c_1, c_2, R_\varrho e^{{\rm i}\beta_3}, R_\varrho e^{{\rm i}\beta_4})$.
Then we see that
$$
X(a,b,c)=\frac 12 X(a^\prime,b^\prime,c^\prime)+\frac 12 X(a^\prime,b^\prime,c^{\prime\prime}))+D
$$
with a diagonal state $D$. Here, the first two summands
are separable {\sf X}-states of rank four by Theorem \ref{rank-four}.

\setlength{\unitlength}{.8 mm}
\begin{figure}
\setlength{\unitlength}{.8 mm}
\begin{center}
\begin{picture}(50,50)
  \qbezier(25.000,0.000)(33.284,0.000)
          (39.142,5.858)
  \qbezier(39.142,5.858)(45.000,11.716)
          (45.000,20.000)
  \qbezier(45.000,20.000)(45.000,28.284)
          (39.142,34.142)
  \qbezier(39.142,34.142)(33.284,40.000)
          (25.000,40.000)
  \qbezier(25.000,40.000)(16.716,40.000)
          (10.858,34.142)
  \qbezier(10.858,34.142)( 5.000,28.284)
          ( 5.000,20.000)
  \qbezier( 5.000,20.000)( 5.000,11.716)
          (10.858,5.858)
  \qbezier(10.858,5.858)(16.716,0.000)
          (25.000,0.000)
\drawline(25,20)(35,37)
\put(35,37){\circle*{1}}
\drawline(25,20)(42,10)
\put(42,10){\circle*{1}}
\drawline(25,20)(38.5,23.5)
\put(38.5,23.5){\circle*{1}}
\drawline(35,37)(42,10)
\put(43,9){$R_\varrho e^{{\rm i}\beta_i}$}
\put(35,38){$R_\varrho e^{{\rm i}\alpha_i}$}
\put(39,24){$c_i$}
\put(28,22){$\phi$}
\put(29,18){$\phi$}
\end{picture}
\end{center}
\caption{$c_i$ is the midpoint of two points on the circle of radius $R_\varrho$.}
\end{figure}
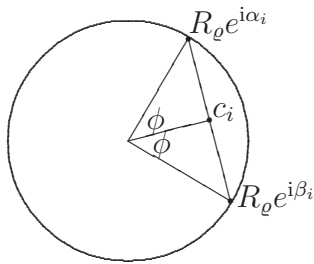

It remains to prove the other cases: $\{i_1,i_2\}$ is $\{1,3\}, \{1,4\}, \{2,3\}, \{2,4\}$ or  $\{3,4\}$.
To do this, we consider the operation on $M_2 \otimes M_2 \otimes M_2$ which interchanges the second and third subsystems. Then
this operation preserves separability and sends $X(a,b,c)$ to
$$
X((a_1,a_3,a_2,a_4),\ (b_1,b_3,b_2,b_4),\ (c_1,c_3,c_2,c_4)).
$$
If $\{i_1,i_2\}=\{1,3\}$ then we take this operation to get the above states. Applying the previous result with $\{i_1,i_2\}=\{1,2\}$
gives rise to the relations identical with those in (ii). For the case $\{i_1,i_2\}=\{1,4\}$, we may consider the operation
which interchange the first and third subsystems. This operation sends $X(a,b,c)$ to the state
$$
X((a_1,b_4,a_3,b_2),\ (b_1,a_4,b_3,a_2),\ (c_1,\bar c_4, c_3,\bar c_2))
$$
for which all the conditions do not change.
For the remaining cases, we use the operations which interchange
$|0\ran$ and $|1\ran$ in the second or third subsystem, and apply the results
when $\{i_1,i_2\}$ is $\{1,2\}$, $\{1,3\}$ and $\{1,4\}$.
This completes the proof.
\end{proof}

A separable {\sf X}-state of rank six may have zero anti-diagonal entries, as we see in the following example:
$$
X({\bf 1},{\bf 1},(1,1,0,0))={1 \over 2}\left[X({\bf 1},{\bf 1},{\bf 1}) + X({\bf 1},{\bf 1},(1,1,-1,-1))\right].
$$
In case that $\{i_1,i_2\}$ is one of $\{1,2\},\ \{3,4\},\ \{1,3\},\ \{2,4\}$, the conditions (ii) and (iii) in Theorem \ref{main-rank-6}
tells us that the relation $c_1c_4=c_2c_3$ actually holds. But, if $\{i_1,i_2\}=\{1,4\}$ or $\{2,3\}$, then this relation
may not hold for separable {\sf X}-states with rank six. Indeed, $\varrho= X({\bf 1},{\bf 1},(1,r,r,1))$ is separable
whenever $-1 \le r\le 1$, but the relation $c_1c_4=c_2c_3$ holds only when $r=\pm 1$.

If an {\sf X}-state $\varrho=X(a,b,c)$ is of rank five, then we have $R_\varrho=r_\varrho$ in Theorem \ref{main-rank-6}, and the relation
$\sqrt{a_ib_i}=r_\varrho$ holds for three of $i=1,2,3,4$. We denote these three by $i_1,i_2,i_3$.
If we take $d,e\in\mathbb R^4$ by $d_i=a_i, e_i=b_i$ for $i=i_1,i_2,i_3$ and take $d_{i_4}, e_{i_4}$
so that the relations (\ref{rank-4}) hold then $X(d,e,c)$ is a separable {\sf X}-state of rank four, and
$\varrho-X(d,e,c)$ is a diagonal state. Therefore, we see that a separable {\sf X}-state of rank five is the sum
of a separable {\sf X}-state of rank four and a diagonal state of rank one or two.
We may go further to show that $\varrho$ has a unique decomposition. To do this, we need the following:

\begin{lemma}\label{kyglhlkh}
Suppose that a separable {\sf X}-state $\varrho=X(a,b,c)$ has a decomposition $\varrho=\sum_k \lambda_k X(d^k,e^k,f^k)+\varrho_D$
into the sum of non-diagonal separable {\sf X}-states of rank four and a diagonal state.
Suppose that $\sqrt {a_{i_0}b_{i_0}}=|c_{i_0}|=|f^k_{i_0}|=1$ for some $i_0$ and all $k$. Then, we have
$\sum_k\lambda_k=1$, $d^k_{i_0}=a_{i_0}$, $e^k_{i_0}=b_{i_0}$ and $f^k_{i_0}=c_{i_0}$ for each $k$.
\end{lemma}

\begin{proof}
We again consider the inequalities (\ref{jgilugjkhkj}), where all the equalities must be identities for $i=i_0$.
By the first identity, we have $f^k_{i_0}=c_{i_0}$ for all $k$. Since $\sum_k \lambda_k f_{i_0}^k = c_{i_0}$,
we have $\sum_k \lambda_k =1$. By the third identity, there exists
$\mu >0$ such that $d^k_{i_0}=\mu e^k_{i_0} = \mu / d^k_{i_0}$ for each $k$, and so, we have $d^k_{i_0}=\sqrt\mu$ for each $k$.
From
$$
1=\sum_k\lambda_k |f^k_{i_0}|^2 = \sum_k\lambda_k d^k_{i_0} e^k_{i_0} =\sqrt\mu \sum_k \lambda_k e^k_{i_0} \le \sqrt\mu b_{i_0}
$$
and
$$
\sqrt{\mu} = \sum_k \lambda_k d_{i_0}^k \le a_{i_0},
$$
we obtain $a_{i_0}=\sqrt{\mu}=d^k_{i_0}$ for each $k$. Similarly, we have $b_{i_0} = e^k_{i_0}$.
\end{proof}

\begin{theorem}
Every non-diagonal separable {\sf X}-state of rank five has a unique decomposition into the sum of a non-diagonal separable {\sf X}-state
of rank four and a diagonal state of rank one or two.
Furthermore, the decomposition into the convex sum of pure product states is also unique.
\end{theorem}

\begin{proof}
Suppose that $\varrho=X(a,b,c)$ is a separable {\sf X}-state of rank five, and
$\sqrt{a_ib_i}=R_\varrho=r_\varrho$ holds for $i=i_1,i_2,i_3$.
Let $\varrho=\sum_k \lambda_k X(d^k,e^k,f^k)+D$ be a decomposition of $\varrho$ into the
sum of non-diagonal separable {\sf X}-states of rank four and a diagonal state.
We may assume that $|c_i|=|f^k_i|=1$ for all $i,k$. Then
we see that $\sum_k \lambda_k X(d^k,e^k,f^k)$ is a single non-diagonal separable {\sf X}-state of rank four
by Lemma \ref{kyglhlkh} and Theorem \ref{rank-four}. Let $\varrho = X_1 + D_1 = X_2 + D_2$ for non-diagonal separable
{\sf X}-states $X_1, X_2$ of rank four and diagonal states $D_1, D_2$. Since $\varrho$ is of rank five,
all entries of diagonal states $D_1, D_2$ are zero except $(i_0,i_0)$ and $(9-i_0,9-i_0)$ positions for
some $i_0$. Hence, $X_1$ and $X_2$ coincide except $(i_0,i_0)$ and $(9-i_0,9-i_0)$ positions.
By (\ref{rank-4}), entries on those positions determined by other diagonal entries.
It follows that $X_1=X_2$ and $D_1=D_2$.

Suppose that $\varrho = \sum_k \lambda_k |\xi_k \ran\lan \xi_k| + \sum_\ell \mu_\ell |\eta_\ell \ran\lan \eta_\ell|$
is a decomposition into pure product states, where every entry of $|\xi_k\ran$ is
nonzero and some entry of $\eta_\ell$ is zero. Taking their {\sf X}-parts,
we get $\varrho = \sum_k \lambda_k X(d^k, e^k, f^k) +D$,
where $X(d^k, e^k, f^k)$ and $D$ are {\sf X}-parts of $|\xi_k \ran\lan \xi_k|$
and $\sum_\ell \mu_\ell |\eta_\ell \ran\lan \eta_\ell|$, respectively.
Since $X(d^k, e^k, f^k)$ is non-diagonal of rank four and $D$ is diagonal,
all $X(d^k, e^k, f^k)$ coincide up to scalar by Lemma \ref{kyglhlkh} and (\ref{rank-4}).
By Theorem \ref{rank-four}, we conclude that each $|\xi_k\ran$ must be one of four vectors
in (\ref{decom-rank1}) for a fixed $|\xi\ran$.
Now, we look at the diagonal part $D=\sum_\ell \mu_\ell |\eta_\ell \ran\lan \eta_\ell|$.
Since $\varrho$ is of rank five, all entries of a diagonal state $D$ are zero
except $(i_0,i_0)$ and $(9-i_0,9-i_0)$ positions for some $i_0$.
This shows that $|\eta_\ell\ran$ must be either $|i\ran |j\ran |k\ran$ or $|\bar i\ran |\bar j\ran |\bar k\ran$
for $i,j,k=0,1$, where $\bar i$ is given by the relation $i+\bar i=1\mod 2$.
This completes the proof of uniqueness of decomposition.
\end{proof}

Now, we have seen that a separable {\sf X}-state $\varrho$ of rank five is located in the simplcial face determined by four product vectors
in (\ref{decom-rank1}) together with $|i\ran |j\ran |k\ran$ and $|\bar i\ran |\bar j\ran |\bar k\ran$.
These six product vectors span a five dimensional space, and determine a $5$-simplex.
If we take five of them then they also span the five dimensional space.
Therefore, they makes a $4$-simplex whose interior points
are still of rank five. This $4$-simplex has an {\sf X}-state if and only if the five choice includes all of product vectors in
(\ref{decom-rank1}).
In short, a separable {\sf X}-state of rank five is decomposed into the sum of six or five pure product states, which include
four pure product states arising from (\ref{decom-rank1}).
Compare with Theorem 4.4 in \cite{ha-kye-multi-unique}.

For a three-qubit state $\varrho$, we consider the number $P_\varrho$ defined by
$$
P_\varrho= \rk\varrho +\rk \varrho^{\Gamma_A} + \rk \varrho^{\Gamma_B}+\rk \varrho^{\Gamma_C}.
$$
It was observed in \cite{abls} and proved in \cite{kiem-kye-na} that
if $P_\varrho\le 28$ then generically there exists no product vector
in the range of $\varrho$ satisfying the range criterion \cite{horo_range}.
In short, a PPT state $\varrho$ with $P_\varrho\le 28$ is entangled with the probability one.
We see that if $\rk\varrho\le 6$ then $P_\varrho\le 28$.
In this contexts, it is not so surprising that our separability characterization
involves an identity; the phase identity.
If $\varrho$ itself and all the partial transposes of $\varrho$ have rank seven
then we still have $P_\varrho = 28$, and so one may suspect if
the phase identity is still necessary for separability.
This is not the case, as we see in the following example:
$$
\begin{aligned}
\varrho &=X({\bf 1},{\bf 1},(1,\textstyle{1 \over 3},\textstyle{{\rm i} \over 3},\textstyle{2-{\rm i} \over 3}))\\
&={1 \over 3}\left[X({\bf 1},{\bf 1},(1,1,1,1)) + X({\bf 1},{\bf 1},(1,{\rm i},-1,-{\rm i})) + X({\bf 1},{\bf 1},(1,-{\rm i},{\rm i},1))\right].
\end{aligned}
$$
This is a separable state, and all the partial transposes have rank seven as well as $\varrho$ itself.
But, $\varrho$ does not satisfy the phase identity: $\theta_2+\theta_3 = \pi \slash 2 \ne \theta_1+\theta_4$.

%%%%%%%%%%%%%%%%%%%%%%%%%%%%%%%%%%%%%%%%%%%%%%%%%%%%%%%%%%%%%%%%%%%%%%%%%%%%%%%%%%%%%%%%%%%%%%%%%%%%%%%%%%%%%%%%%%
%%%%%%%%%%%%%%%%%%%%%%%%%%%%%%%%%%%%%%%%%%%%%%%%%%%%%%%%%%%%%%%%%%%%%%%%%%%%%%%%%%%%%%%%%%%%%%%%%%%%%%%%%%%%%%%%%%
%%%%%%%%%%%%%%%%%%%%%%%%%%%%%%%%%%%%%%%%%%%%%%%%%%%%%%%%%%%%%%%%%%%%%%%%%%%%%%%%%%%%%%%%%%%%%%%%%%%%%%%%%%%%%%%%%%
%%%%%%%%%%%%%%%%%%%%%%%%%%%%%%%%%%%%%%%%%%%%%%%%%%%%%%%%%%%%%%%%%%%%%%%%%%%%%%%%%%%%%%%%%%%%%%%%%%%%%%%%%%%%%%%%%%
%%%%%%%%%%%%%%%%%%%%%%%%%%%%%%%%%%%%%%%%%%%%%%%%%%%%%%%%%%%%%%%%%%%%%%%%%%%%%%%%%%%%%%%%%%%%%%%%%%%%%%%%%%%%%%%%%%
\section{optimal decompositions}

A decomposition of a separable state $\varrho$ into the sum of pure product states is said to be optimal when the number of
pure product states is minimal. This minimal number is also called the length $\ell(\varrho)$ of $\varrho$.
%The length $\ell(\varrho)$ of a separable state $\varrho$ is defined \cite{DiVin} by the minimum natural number $k$
%with which the decomposition of $\varrho$ into the convex sum of $k$ pure product states is possible.
The length of an $m\otimes n$ bi-partite separable state does not exceed $(mn)^2$ \cite{horo_range}.
It was shown in \cite{DiVin} that the length may be strictly greater than the rank of the state.
Chen and Djokovi\'c \cite{chen_dj_2xd} showed that the length of a $2\otimes 3$ separable state $\varrho$
is equal to the maximum of $\rk\varrho$ and $\rk\varrho^\Gamma$. They also showed
\cite{chen_dj_semialg} that the length of an $m\otimes n$ separable state
may exceed $mn$ when $(m-2)(n-2)>1$. Later, $3\otimes 3$ and $2\otimes 4$ separable states with length $10$
have been constructed in \cite{{ha-kye-sep_2x4},{ha-kye-sep-face}}.
In the previous section, we have shown that a three-qubit separable {\sf X}-state of rank $5$ has length $5$ or $6$.
We show in this section that the lengths $\ell (\varrho)$ of three-qubit separable {\sf X}-states $\varrho$ of rank six are
equal to the maximum
$$
\Gamma(\varrho)= \max\{\rk\varrho,\ \rk \varrho^{\Gamma_A},\ \rk \varrho^{\Gamma_B},\ \rk \varrho^{\Gamma_C}\}
$$
of the ranks, except for the case
\begin{equation}\label{exceptional}
\Gamma(\varrho)=7, \qquad a_1 a_4 \ne a_2 a_3,\qquad b_1 b_4 \ne b_2 b_3,
\end{equation}
for which we have $\ell(\varrho)=8$. Note that $\ell(\varrho)\ge \Gamma(\varrho)$ in general.
In the course of discussion, we get an alternative proof for Theorem \ref{main-rank-6}.

From now on, we suppose that $\varrho=X(a,b,c)$ is a three-qubit non-diagonal separable {\sf X}-state with rank $\le 6$.
By the PPT condition $\sqrt{a_ib_i}\ge R_\varrho$, we see that
$\sqrt{a_ib_i}=R_\varrho=|c_i|$ for $i\in\{i_1,i_2\}$. Without loss of generality, we
may assume that $R_\varrho=1$. We consider the case of $\{i_1,i_2\}=\{1,2\}$.
Note that $\varrho$ has two kernel vectors;
$$
(c_1,0,0,0,0,0,0,0,-a_1)^\ttt,\qquad
(0,c_2,0,0,0,0,-a_2,0)^\ttt.
$$
If a product vector $|x\ran\ot |y\ran\ot |z\ran\in\mathbb C^2\otimes\mathbb C^2\otimes\mathbb C^2$
belongs to the range of $\varrho$,
%. Then we have the relation
%$$
%|x\ran y\ran z\ran \in \im\varrho,\quad
%|\bar x\ran y\ran z\ran \in \im\varrho^{\Gamma_A},\quad
%|x\ran \bar y\ran z\ran \in \im\varrho^{\Gamma_B},\quad
%|x\ran y\ran \bar z\ran \in \im\varrho^{\Gamma_C}.
%$$
%We note that $\varrho$ has two kernel vectors
%$$
%(c_1,0,0,0,0,0,0,0,-a_1)^\ttt,\qquad
%(0,c_2,0,0,0,0,-a_2,0)^\ttt,
%$$
%and $\varrho^{\Gamma_C}$ also has two kernel vectors
%$$
%(c_2,0,0,0,0,0,0,0,-a_1)^\ttt,\qquad
%(0,c_1,0,0,0,0,-a_2,0)^\ttt.
%$$
%\begin{equation}\label{kye_eq_1}
%\begin{aligned}
%\bar c_1x_0y_0z_0&=a_1x_1y_1z_1\\
%\bar c_2 x_0y_0z_1&=a_2x_1y_1z_0\\
%\bar c_2 x_0y_0\bar z_0&=a_1x_1y_1\bar z_1\\
%\bar c_1 x_0y_0\bar z_1&=a_2 x_1y_1\bar z_0
%\end{aligned}
%\end{equation}
then we have the equations
\begin{equation}\label{kye_eq_1}
\bar c_1x_0y_0z_0=a_1x_1y_1z_1,\qquad
\bar c_2 x_0y_0z_1=a_2x_1y_1z_0,
\end{equation}
where $|x\ran=(x_0,x_1)^\ttt$, $|y\ran=(y_0,y_1)^\ttt$ and $|z\ran=(z_0,z_1)^\ttt$.
It is routine to check that a solution with a zero entry is one of the following:
\begin{equation}\label{zero-entry}
\begin{aligned}
|\xi_x\ran :=(x_0,x_1)^\ttt\otimes (1,0)^\ttt \otimes (0,1)^\ttt &=(0,x_0,0,0,0,x_1,0,0)^\ttt,\\
|\eta_x\ran :=(x_0,x_1)^\ttt\otimes (0,1)^\ttt \otimes (1,0)^\ttt &=(0,0,x_0,0,0,0,x_1,0)^\ttt,\\
|\xi_y\ran :=(1,0)^\ttt\otimes (y_0,y_1)^\ttt \otimes (0,1)^\ttt &=(0,y_0,0,y_1,0,0,0,0)^\ttt,\\
|\eta_y\ran :=(0,1)^\ttt\otimes (y_0,y_1)^\ttt \otimes (1,0)^\ttt &=(0,0,0,0,y_0,0,y_1,0)^\ttt,\\
|\xi_z\ran :=(1,0)^\ttt\otimes (0,1)^\ttt \otimes (z_0,z_1)^\ttt &=(0,0, z_0,z_1,0,0,0,0)^\ttt,\\
|\eta_z\ran :=(0,1)^\ttt\otimes (1,0)^\ttt \otimes (z_0,z_1)^\ttt &=(0,0,0,0, z_0,z_1,0,0)^\ttt.
\end{aligned}
\end{equation}
Now, we proceed to look for solutions which have no zero entries. In this case, we may put $x_0=y_0=z_0=1$,
and see that the solution is of the form
$$
x_1=b_1^\ha c_1^{-\ha}\alpha^{-1},\quad y_1=\pm b_2^\ha c_2^{-\ha}\alpha,\quad
z_1 =\pm b_1^{\ha}a_2^\ha c_1^{-\ha}c_2^\ha,
$$
with a nonzero complex number $\alpha$. Therefore, the solutions are given by
$$
\begin{aligned}
|\xi_\alpha^\pm\ran
:=&(1,b_1^\ha c_1^{-\ha}\alpha^{-1})^\ttt
   \otimes (1,\pm b_2^\ha c_2^{-\ha}\alpha)^\ttt
   \otimes (1,\pm b_1^{\ha}a_2^\ha c_1^{-\ha}c_2^\ha)^\ttt\\
=&(\alpha c_1^\ha,b_1^\ha)^\ttt \otimes (\alpha^{-1}c_2^\ha, \pm b_2^\ha)^\ttt
    \otimes (a_1^\ha c_1^\ha, \pm a_2^\ha c_2^\ha)^\ttt\\
=&(a_1^\ha c_1c_2^\ha,
   \pm a_2^\ha c_1^\ha c_2,
   \pm \alpha a_1^\ha b_2^\ha c_1,
   \alpha c_1^\ha c_2^\ha,
   \alpha^{-1} c_1^\ha c_2^\ha,
   \pm \alpha^{-1} a_2^\ha b_1^\ha c_2,
   \pm b_2^\ha c_1^\ha,
   b_1^\ha c_2^\ha)^\ttt,
\end{aligned}
$$
up to scalar multiplication.

Hence, the only possible decomposition of $\varrho$ into pure product states is of the form
$\varrho=\varrho_1+\varrho_0$, where
$$
\varrho_1
=
\sum_{i\in I} p_i|\xi_{\alpha_i}^+\ran\lan\xi_{\alpha_i}^+|
+
\sum_{j\in J} q_j|\xi_{\beta_j}^-\ran\lan\xi_{\beta_j}^-|
$$
and $\varrho_0$ is a separable state which is decomposed into the sum of pure states in (\ref{zero-entry}).
%The $(1,1),(8,8)$-entries of $\varrho_0$ are zero among diagonal entries, and $(i,j)$-entries of $\varrho_0$ are zero except for
%$$
%(i,j)=(3,7),(7,3),(2,6),(6,2),(5,7),(7,5),(2,4),(4,2),(5,6),(6,5),(3,4),(4,3)
%$$
%among non-diagonal entries.
We compare $(1,j)$-entries and $(4,5),(4,6)$-entries of $\varrho$ and $\varrho_1+\varrho_0$ to get the relations
\begin{equation}\label{kye_eq1}
\sum_{i\in I} p_i=\sum_{j\in J}q_j=\ha,\qquad
\sum_{i\in I}p_i\alpha_i
=\sum_{i\in I}p_i\alpha_i^{-1}
=\sum_{j\in J}q_j\beta_j
=\sum_{j\in J}q_j\beta_j^{-1}=0
\end{equation}
and
\begin{equation}\label{kye_eq2}
\sum_{i\in I}p_i\alpha_i\bar\alpha_i^{-1}=\sum_{j\in J}q_j\beta_j\bar\beta_j^{-1}=\ha c_4,
\end{equation}
respectively.
These conditions imply that
$$
\varrho_1=
\left(\begin{matrix}
a_1 &\cdot&\cdot&\cdot&\cdot&\cdot&\cdot  &c_1\\
\cdot&a_2 &\cdot&\cdot&\cdot&\cdot &c_2 &\cdot \\
\cdot&\cdot& a_1b_2A_+ &a_1^\ha b_2^\ha c_1^\ha \bar c_2^\ha A_- &\cdot &c_1\bar c_2 c_4 &\cdot&\cdot \\
\cdot&\cdot&a_1^\ha b_2^\ha \bar c_1^\ha c_2^\ha A_-& A_+ &c_4 &\cdot&\cdot&\cdot \\
\cdot&\cdot&\cdot& \bar c_4 &B_+ &a_2^\ha b_1^\ha c_1^\ha \bar c_2^\ha B_- &\cdot&\cdot \\
\cdot&\cdot& \bar c_1 c_2\bar c_4 &\cdot&a_2^\ha b_1^\ha \bar c_1^\ha c_2^\ha B_- &a_2b_1B_+ &\cdot&\cdot\\
\cdot&\bar c_2 &\cdot&\cdot&\cdot&\cdot& b_2 &\cdot \\
\bar c_1 &\cdot&\cdot&\cdot&\cdot&\cdot&\cdot& b_1
\end{matrix}\right),
$$
%$$
%\varrho_1=
%X((a_1,a_2, a_1b_2A, A), (b_1,b_2, a_2b_1B, B), (c_1,c_2,c_1\bar c_2c_4,c_4)),
%$$
%and $\varrho-\varrho_1$ is of the form $\varrho_0$,
where
\begin{equation}\label{exppr_AB}
A_\pm =\sum_{i\in I} p_i|\alpha_i|^2\pm \sum_{j\in J} q_j|\beta_j|^2,\qquad
B_\pm =\sum_{i\in I} p_i|\alpha_i|^{-2}\pm \sum_{j\in J} q_j|\beta_j|^{-2}.
\end{equation}
From this, we see that the numbers $A_+$ and $B_+$ satisfy the relations
\begin{equation}\label{decom-AB}
a_1b_2A_+ \le a_3,\qquad A_+ \le a_4,\qquad a_2b_1B_+ \le b_3,\qquad B_+ \le b_4,
\end{equation}
and $1,2,7,8$-th rows and columns of $\varrho_0$ are zero.
%, that is, only $\xi_z$ and $\eta_z$ in (\ref{zero-entry}) involve.

In the above discussion, we actually have shown that $c_1\bar c_2c_4=c_3$, or equivalently $c_1c_4=c_2c_3$. Furthermore, we see that
$$
a_1b_2b_3a_4\ge a_1b_2(a_2b_1B_+)A_+=A_+B_+\ge 1,
$$
where the last inequality follows from the Cauchy-Schwarz inequality.
We also have $b_1a_2a_3b_4\ge 1$ similarly. Therefore, we have an alternative proof for the necessity part of
Theorem \ref{main-rank-6}.

\begin{theorem}
Suppose that $\varrho$ is a three-qubit non-diagonal separable {\sf X}-state of rank $6$.
We have $\Gamma(\varrho)=\ell(\varrho)$ except for the case of {\rm (\ref{exceptional})}.
When {\rm (\ref{exceptional})} holds, we have $\ell(\varrho)=8$.
\end{theorem}

\begin{proof}
We may assume that $R_\varrho=1$ as before. We may also assume that $\{i_1,i_2\}=\{1,2\}$
by the argument in the proof of Theorem \ref{main-rank-6}.
First, let us show that $\ell(\varrho) \le 8$.
Since $|c_4|\le R_\varrho = 1$, we can take $0\le t\le 1$ and $\theta$ so that $c_4=(2t-1)e^{2{\rm i}\theta}$.
If we put
\begin{equation}\label{rs}
\begin{gathered}
\alpha_1=\sqrt{r}e^{\rm i \theta},\quad
  \alpha_2=\sqrt{r}e^{\rm i \theta}(-t+{\rm i}\sqrt{1-t^2}),\quad
  \alpha_3=\sqrt{r}e^{\rm i \theta}(-t-{\rm i}\sqrt{1-t^2}),\\
\beta_1=\sqrt{s}e^{\rm i \theta},\quad
  \beta_2=\sqrt{s}e^{\rm i \theta}(-t+{\rm i}\sqrt{1-t^2}),\quad
  \beta_3=\sqrt{s}e^{\rm i \theta}(-t-{\rm i}\sqrt{1-t^2}),
\end{gathered}
\end{equation}
and
$$
p_1={t \over 2(1+t)}, \quad p_2=p_3={1 \over 4(1+t)}, \quad q_i=p_i,
$$
then (\ref{kye_eq1}) and (\ref{kye_eq2}) are satisfied. Furthermore, we have
$$
A_\pm={1 \over2}(r \pm s), \qquad B_\pm={1 \over2}(r^{-1} \pm s^{-1}).
$$

In order to prove $\ell(\varrho) \le 8$, it suffices to show that there exist $z$ and $z^\prime$ in $\mathbb C^2$
satisfying
\begin{equation}\label{8-decom}
\varrho = \varrho_1 + |\xi_z \ran \lan \xi_z | + |\eta_{z'}\ran \lan \eta_{z'}|.
\end{equation}
For this purpose, we put
$$
\begin{aligned}
z &=\frac 1{\sqrt 2}\left(\lambda^\ha a_1^\ha b_2^\ha c_1^\ha |r-s|^\ha,\ \pm \lambda^{-\ha}c_2^\ha |r-s|^\ha\right)^\ttt,\\
z'&=\frac 1{\sqrt 2}\left(\mu^{-\ha} c_1^\ha |r^{-1}-s^{-1}|^\ha,\ \mp \mu^\ha a_2^\ha b_1^\ha c_2^\ha |r^{-1}-s^{-1}|^\ha\right)^\ttt,
\end{aligned}
$$
where $\lambda, \mu>0$ and the signs are determined whether $s>r$ or $s<r$ in order.
We compare $\{3,4\}$ and $\{5,6\}$ principle submatrices of (\ref{8-decom}), to get
\begin{equation}\label{hhhhh-1}
\begin{aligned}
\begin{pmatrix}
a_3&0\\
0&a_4
\end{pmatrix}
=&{1 \over 2}
\begin{pmatrix}
a_1b_2(r+s) & a_1^\ha b_2^\ha c_1^\ha \bar c_2^\ha(r-s)\\
a_1^\ha b_2^\ha \bar c_1^\ha c_2^\ha(r-s)&r+s
\end{pmatrix}\\
& \phantom{XXX}+{1 \over 2}
\begin{pmatrix}
a_1b_2\lambda |r-s|& a_1^\ha b_2^\ha c_1^\ha \bar c_2^\ha(s-r)\\
a_1^\ha b_2^\ha \bar c_1^\ha c_2^\ha(s-r)&\lambda^{-1}|r-s|
\end{pmatrix},
\end{aligned}
\end{equation}
and
\begin{equation}\label{hhhhh-2}
\begin{aligned}
\begin{pmatrix}
b_4&0\\
0&b_3
\end{pmatrix}
=&{1 \over 2}
\begin{pmatrix}
r^{-1}+s^{-1} & a_2^\ha b_1^\ha c_1^\ha \bar c_2^\ha(r^{-1}-s^{-1})\\
a_2^\ha b_1^\ha \bar c_1^\ha c_2^\ha(r^{-1}-s^{-1})&a_2b_1(r^{-1}+s^{-1})
\end{pmatrix}\\
& \phantom{XXX}+{1 \over 2}
\begin{pmatrix}
\mu^{-1}|r^{-1}-s^{-1}|& a_2^\ha b_1^\ha c_1^\ha \bar c_2^\ha(s^{-1}-r^{-1})\\
a_2^\ha b_1^\ha \bar c_1^\ha c_2^\ha(s^{-1}-r^{-1})&a_2b_1\mu|r^{-1}-s^{-1}|
\end{pmatrix}.
\end{aligned}
\end{equation}
Comparing the diagonal entries of (\ref{hhhhh-1}) and (\ref{hhhhh-2}),
we see that the existence of $z$ and $z^\prime$ satisfying (\ref{8-decom}) is guaranteed if we
can take $r,s>0$ satisfying
\begin{equation}\label{kkkk}
\begin{aligned}
(r-s)^2&=(r+s-2b_1a_2a_3)(r+s-2a_4),\\
(r-s)^2&=(r+s-2a_1b_2b_3 rs)(r+s-2b_4 rs).
\end{aligned}
\end{equation}
By putting $x=r+s$ and $y=rs$, this is equivalent to
\begin{equation}\label{eq-xy}
(b_1a_2a_3+a_4)x-2y=2b_1a_2a_3a_4, \qquad (a_1b_2b_3+b_4)x-2a_1b_2b_3b_4y=2.
\end{equation}

If $a_3b_3>1$ or $a_4b_4>1$ then the solution is given by
$$
x={2(a_3b_3a_4b_4-1) \over b_4(a_3b_3 -1)+a_1b_2b_3(a_4b_4-1)},\quad
y={a_4(a_3b_3-1)+b_1a_2a_3(a_4b_4-1) \over b_4(a_3b_3 -1)+a_1b_2b_3(a_4b_4-1)}.
$$
We note that both the above $x$ and $y$ are nonnegative. Since
the discriminant
$$
x^2-4y={4(a_3b_3-1)(a_4b_4-1)(a_1b_2b_3a_4-1)(b_1a_2a_3b_4-1) \over (a_3b_3b_4 + a_1b_2b_3a_4b_4 - b_4 -a_1b_2b_3)^2}
$$
is nonnegative, we can take $r$ and $s$ satisfying (\ref{kkkk}).
If $a_3b_3=1$ and $a_4b_4=1$, then $r=s=a_4$ satisfies (\ref{kkkk}),
and this completes the proof of $\ell(\varrho) \le 8$.

We classify non-diagonal separable {\sf X}-states of rank six by the number $\Gamma(\varrho)$ as follows:
\begin{enumerate}
\item[(i)]
$\Gamma(\varrho)=6$; $\sqrt{a_ib_i} = R_\varrho > r_\varrho$ for $i=3,4$
   or $\sqrt{a_ib_i} > R_\varrho = r_\varrho$ for $i=3,4$,
\item[(ii)]
$\Gamma(\varrho)=7$; $\sqrt{a_4b_4} > \sqrt{a_3b_3} = R_\varrho > r_\varrho$ or
   $\sqrt{a_3b_3} > \sqrt{a_4b_4} = R_\varrho > r_\varrho$,
\item[(iii)]
$\Gamma(\varrho)=8$; $\sqrt{a_ib_i} > R_\varrho > r_\varrho$ for $i=3,4$.
\end{enumerate}
%It is easy to check that
%$$
%\Gamma(\varrho_A)=8,\quad \Gamma(\varrho_B)=6,\quad \Gamma(\varrho_C)=7,\quad \Gamma(\varrho_D)=6.
%$$
%Combining the preceding paragraph, it is immediate that $\ell(\varrho_A)=8$.
We retain our assumption $R_\varrho=1$. If $\sqrt{a_ib_i} = R_\varrho=1$ for $i=3,4$
then we have $a_1b_2b_3a_4=b_1a_2a_3b_4=1$ by $\Delta_\varrho\ge 1$.
Taking $r=s=a_4$ in (\ref{rs}), we have
$$
A_+=a_4,\quad a_1b_2A_+=a_3,\quad B_+=b_4,\quad a_2b_1B_+=b_3,\quad A_-=0=B_-,
$$
which shows $\ell(\varrho)=6$. In case of $R_\varrho=r_\varrho=1$, we write $c_4=e^{2{\rm i}\theta}$ and put
$$
\alpha_1=\sqrt{r}e^{\rm i \theta},\quad
  \alpha_2=-\sqrt{r}e^{\rm i \theta}, \quad
\beta_1=\sqrt{s}e^{\rm i \theta},\quad
  \beta_2=-\sqrt{s}e^{\rm i \theta}
$$
and $p_1=p_2=q_1=q_2={1 \over 4}$,
then (\ref{kye_eq1}) and (\ref{kye_eq2}) are satisfied.
The similar argument of the preceding paragraph establishes $\ell(\varrho)=6$.
This proves that if $\Gamma(\varrho)=6$ then $\ell(\varrho)=6$. If
$\Gamma(\varrho)=8$ then we also have $8= \Gamma(\varrho)\le \ell(\varrho)\le 8$, and so $\ell(\varrho)=8$.

Finally, we consider the case when $\Gamma(\varrho)=7$. We have $|I|\ge 2$ and $|J|\ge 2$ by (\ref{kye_eq1}).
When $|I|=2$, solving algebraic equations
$$
p_1+p_2=1/2, \qquad p_1\alpha_1 + p_2\alpha_2 = 0 = p_1 \alpha_1^{-1}+p_2 \alpha_2^{-1}
$$
yields the solutions
$\alpha_1=-\alpha_2$ and $p_1=p_2=1/4$.
We also see that $\alpha_1\bar\alpha^{-1}_1=c_4$ by (\ref{kye_eq2}), and so
$|c_3|=|c_4|=1$ which implies $R_\varrho=r_\varrho$. Therefore, we have $|I|\ge 3$.
Similarly, we also have $|J|\ge 3$.
Note that $(\varrho_1)_{11}(\varrho_1)_{44}=(\varrho_1)_{22}(\varrho_1)_{33}$ and
$(\varrho_1)_{55}(\varrho_1)_{88}=(\varrho_1)_{66}(\varrho_1)_{77}$.
Therefore, the condition (\ref{exceptional}) implies that $\ell(\varrho)=8$.

Now, it remains to consider the case when $\Gamma(\varrho)=7$ but (\ref{exceptional}) does not hold.
We first consider the case $a_4b_4>a_3b_3=1$. If $a_1a_4=a_2a_3$ then $a_1b_2a_4=a_3$. In this case,
%If $a_1a_4 \ne a_2a_3$ and $b_1 b_4 \ne b_2 b_3$, then both $\eta_z$ and $\xi_z$ are nonzero.
%Therefore, we have $\ell(\varrho_C)=8$ when (\ref{exceptional}).
%If $a_1a_4 = a_2a_3$ and $b_1b_4 = b_2b_3$, then we have $a_4 b_4 = (b_1a_2a_3)(a_1b_2b_3)=a_3b_3$.
%Therefore, we have
%$$
%a_1a_4 = a_2a_3,\ b_1b_4 \neq b_2b_3,\qquad {\text{\rm or}}\qquad a_1a_4 \neq a_2a_3,\ b_1b_4 = b_2b_3
%$$
%We may assume without the loss of generality that $\sqrt{a_4b_4} > \sqrt{a_3b_3} = R_\varrho > r_\varrho$ and
taking $r=s=a_4$ in (\ref{rs}) yields
$$
A_+=a_4,\quad a_1b_2A_+=a_3,\quad B_+=1/a_4,\quad a_2b_1B_+=b_3,\quad A_-=0=B_-.
$$
Therefore, we see that $\varrho-\varrho_1$ is a diagonal state of rank one since $1/a_4<b_4$, and
conclude that $\ell(\varrho)=7$.
If $b_1b_4=b_2b_3$ then $b_1a_2b_4=b_3$, and so we may take $r=s=1/b_4$ to see $\ell(\varrho)=7$.
The case $a_3b_3>a_4b_4=1$ can be done similarly.
%We get a decomposition $\varrho=\varrho_1 + |\eta_z\ran\lan\eta_z|$ with $z=((b_4-1/a_4)^\ha,0)^\ttt$, hence $\ell(\varrho_C)=7$.
%The second case can be done in the same way.
\end{proof}

In the case of (ii) in the above proof, the condition $\sqrt{a_ib_i}=R_\varrho$ for $i=1,2,3$ implies that
$$
a_1a_4=a_2a_3\ \Longleftrightarrow\ \sqrt[4]{a_1b_2b_3a_4}=R_\varrho,\qquad
b_1b_4=b_2b_3\ \Longleftrightarrow\ \sqrt[4]{b_1a_2a_3b_4}=R_\varrho.
$$
This is also the case when $\sqrt{a_ib_i}=R_\varrho$ for $i=1,2,4$. Therefore, the last two conditions in
(\ref{exceptional}) may be replaced by
$\sqrt[4]{a_1b_2b_3a_4}>R_\varrho$ and $\sqrt[4]{b_1a_2a_3b_4}>R_\varrho.$

%%%%%%%%%%%%%%%%%%%%%%%%%%%%%%%%%%%%%%%%%%%%%%%%%%%%%%%%%%%%%%%%%%%%%%%%%%%%%%%%%%%%%%%%%%%%%%%%%%%%%%%%%%%%%%%%%%
%%%%%%%%%%%%%%%%%%%%%%%%%%%%%%%%%%%%%%%%%%%%%%%%%%%%%%%%%%%%%%%%%%%%%%%%%%%%%%%%%%%%%%%%%%%%%%%%%%%%%%%%%%%%%%%%%%
%%%%%%%%%%%%%%%%%%%%%%%%%%%%%%%%%%%%%%%%%%%%%%%%%%%%%%%%%%%%%%%%%%%%%%%%%%%%%%%%%%%%%%%%%%%%%%%%%%%%%%%%%%%%%%%%%%
%%%%%%%%%%%%%%%%%%%%%%%%%%%%%%%%%%%%%%%%%%%%%%%%%%%%%%%%%%%%%%%%%%%%%%%%%%%%%%%%%%%%%%%%%%%%%%%%%%%%%%%%%%%%%%%%%%
%%%%%%%%%%%%%%%%%%%%%%%%%%%%%%%%%%%%%%%%%%%%%%%%%%%%%%%%%%%%%%%%%%%%%%%%%%%%%%%%%%%%%%%%%%%%%%%%%%%%%%%%%%%%%%%%%%
%\section{Discussion}

\section{Conclusion}

In this note, we gave separability criteria in terms of
diagonal and anti-diagonal entries to detect three-qubit
entanglement, which depends on phases of anti-diagonal entries.
The main criterion is the inequality
$$
\Delta_\varrho\ge r_\varrho\sqrt{1+|\sin\phi_\varrho/2|}
$$
in terms of the phase difference $\phi_\varrho$. This criterion is strong enough to detect
PPT entanglement of nonzero volume.
Anti-diagonal phases play a role in general to determine positivity
of Hermitian matrices. For example, if we consider the $n\times n$
matrix $[a_{i,j}]$ whose entries are all $1$ except for $a_{1,n}=
\bar a_{n,1}=e^{{\rm i}\theta}$ with $n\ge 3$, then it is positive only
when $\theta=0$. But, they play no role for positivity of {\sf
X}-shaped Hermitian matrices. This means that criterion for
positivity with diagonal and anti-diagonal entries depends only on
the magnitudes of entries.

This is also the case for bi-separability and full bi-separability of multi-qubit {\sf X}-states.
A multipartite
state is called an $(S|T)$ bi-separable if it is separable as a
bi-partite state according to the bi-partition $(S|T)$ for systems.
It is called fully bi-separable if it is $(S|T)$ bi-separable for
every bi-partition $(S|T)$, and bi-separable if it is a mixture of
$(S|T)$ bi-separable states through bi-partitions $(S|T)$ for
systems. Characterization of bi-separability and full bi-separability of multi-qubit {\sf X}-states
depends only on the magnitudes of entries, and is free from the phases \cite{{han_kye_EW},{Rafsanjani}}.
In fact, bi-separability and full bi-separability of multi-qubit {\sf X}-states are equivalent to the
corresponding notions of positivity of partial transposes
\cite{han_kye_EW}. Nevertheless, we have shown that phases of
anti-diagonal entries play a crucial role to determine full
separability of three-qubit {\sf X}-states. In other words,
detecting entanglement with the PPT property depends on the
anti-diagonal phases.
It was shown in \cite{han_kye_tri,han_kye_GHZ} that the anti-diagonal phases also play a role to characterize
three-qubit {\sf X}-shaped entanglement witnesses.

Because our main criterion depends on the diagonal and anti-diagonal entries, it is very natural
to ask  for which {\sf X}-states it gives a sufficient condition for separability. We have shown that
this is the case when the anti-diagonal entries of {\sf X}-states share a common magnitude.
In some extreme cases, the phase difference must be zero, or equivalently
the anti-diagonal entries must satisfy the phase identity
$$
\theta_1+\theta_4=\theta_2+\theta_3 \mod 2\pi
$$
for separability. This is the case when the rank is less than or equal to six. Actually,
we characterize separability of three-qubit {\sf X}-states with rank $\le 6$, using the phase identity.
As a byproduct, we found examples of three-qubit separable states whose length is strictly greater than
the rank of every kind of partial transpose.

It would be interesting to find analogous identities for general multi-qubit cases and determine to what extent these identities give rise to a separability criterion. Finally, it would be also nice to have a calculable characterization of separability of general multi-qubit {\sf X}-shaped states.

%%%%%%%%%%%%%%%%%%%%%%%%%%%%%%%%%%%%%%%%%%%%%%%%%%%%%%%%%%%%%%%%%%%%%%%%%%%%%%%%%%%%%%%%%%%%%%%%%%%%%%%%%%%%%%%%%%
%%%%%%%%%%%%%%%%%%%%%%%%%%%%%%%%%%%%%%%%%%%%%%%%%%%%%%%%%%%%%%%%%%%%%%%%%%%%%%%%%%%%%%%%%%%%%%%%%%%%%%%%%%%%%%%%%%
%%%%%%%%%%%%%%%%%%%%%%%%%%%%%%%%%%%%%%%%%%%%%%%%%%%%%%%%%%%%%%%%%%%%%%%%%%%%%%%%%%%%%%%%%%%%%%%%%%%%%%%%%%%%%%%%%%
%%%%%%%%%%%%%%%%%%%%%%%%%%%%%%%%%%%%%%%%%%%%%%%%%%%%%%%%%%%%%%%%%%%%%%%%%%%%%%%%%%%%%%%%%%%%%%%%%%%%%%%%%%%%%%%%%%
%%%%%%%%%%%%%%%%%%%%%%%%%%%%%%%%%%%%%%%%%%%%%%%%%%%%%%%%%%%%%%%%%%%%%%%%%%%%%%%%%%%%%%%%%%%%%%%%%%%%%%%%%%%%%%%%%%


\begin{thebibliography}{99}


\bibitem{abls}
A. Ac\'in, D. Bru\ss, M. Lewenstein and A. Sapera,
\it Classification of mixed three-qubit states,
\rm Phys. Rev. Lett. \bf 87 \rm (2001), 040401.

\bibitem{alfsen}
E. Alfsen and F. Shultz,
\it Unique decompositions, faces, and automorphisms of separable states,
\rm J. Math. Phys. \bf 51 \rm (2010), 052201.

\bibitem{chen_dj_2xd}
L. Chen and D. \v Z. Djokovi\'c,
\it Qubit-qudit states with positive partial transpose,
\rm Phys. Rev. A {\bf 86} (2012), 062332.
%arXiv:1210.0111

\bibitem{chen_dj_semialg}
L. Chen and D. \v Z. Djokovi\'c,
\it Dimensions, lengths and separability in finite-dimensional quantum systems,
\rm J. Math. Phys. {\bf 54} (2013), 022201.
%arXiv:1206.3775

\bibitem{chen_dj_rank4}
L. Chen and D. \v Z. Djokovi\'c,
\it Separability problem for multipartite states of rank at most 4,
\rm J. Phys. A: Math. Theor. {\bf 46} (2013), 275304.

\bibitem{chen_full_sep}
L. Chen and D. \v Z. Djokovi\'c,
\it Boundary of the set of separable states,
%\rm preprint. arXiv:1404.0738.
\rm Proc. R. Soc. A  {\bf 471} (2015), 20150102.

\bibitem{choi-ppt}
M.-D. Choi,
\it Positive linear maps,
\rm Operator Algebras and Applications (Kingston, 1980), pp.
583--590, Proc. Sympos. Pure Math. Vol 38. Part 2, Amer. Math. Soc., 1982.

\bibitem{DiVin}
D. P. DiVincenzo, B. M. Terhal, and A. V. Thapliyal,
\it Optimal decompositions of barely separable states,
\rm J. Mod. Opt. {\bf 47} (2000), 377--385.
%arXiv: quant-ph:9904005

\bibitem{dur_multi}
W. D\" ur and J. I. Cirac,
\it Classification of multiqubit mixed states: Separability and distillability properties,
\rm Phys. Rev. A {\bf 61} (2000), 042314.

\bibitem{dur}
W. D\" ur, J. I. Cirac and R. Tarrach, \it Separability and
Distillability of Multiparticle Quantum Systems, \rm Phys. Rev Lett.
{\bf 83} (1999), 3562--3565.

\bibitem{guhne_pla_2011}
O. G\"uhne,
\it Entanglement criteria and full separability of multi-qubit quantum states,
\rm Phys. Lett. A {\bf 375} (2011), 406--410.

\bibitem{guhne10}
O. G\" uhne and M. Seevinck,
\it Separability criteria for genuine multiparticle entanglement,
\rm New J. Phys. {\bf 12} (2010), 053002.

\bibitem{guhne_survey}
O. G\" uhne and G. Toth,
\it Entanglement detection,
\rm Phys. Rep. \bf 474 \rm (2009), 1--75.

\bibitem{ha-kye-sep_2x4}
K.-C. Ha and S.-H. Kye,
\it Geometry for separable states and construction of entangled states with positive partial transposes,
\rm Phys. Rev. A {\bf 88} (2013), 024302.
%arXiv:1307.1362

\bibitem{ha-kye-sep-face}
K.-C. Ha and S.-H. Kye,
\it Separable states with unique decompositions,
\rm Commun. Math. Phys. {\bf 328} (2014), 131--153.

\bibitem{ha-kye-multi-unique}
K.-C. Ha and S.-H. Kye,
\it Multi-partite separable states with unique decompositions and construction of
three qubit entanglement with positive partial transpose,
\rm J. Phys. A: Math. Theor. {\bf 48} (2015), 045303.
%arXiv:1402.5813

\bibitem{han_kye_tri}
K. H. Han and S.-H, Kye,
\it Various notions of positivity for bi-linear maps and applications to tri-partite entanglement,
\rm J. Math. Phys. {\bf 57} (2016), 015205.

\bibitem{han_kye_EW}
K. H. Han and S.-H, Kye,
\it Construction of multi-qubit optimal genuine entanglement witnesses,
\rm J. Phys. A: Math. Theor. {\bf 49} (2016), 175303.
%\rm arXiv:1510.03620

\bibitem{han_kye_GHZ}
K. H. Han and S.-H, Kye,
\it Separability of three qubit Greenberger-Horne-Zeilinger diagonal states,
\rm J. Phys. A: Math. Theor. {\bf 50} (2017), 145303.

\bibitem{horo_range}
P. Horodecki,
\it Separability criterion and inseparable mixed states with positive partial transposition,
\rm Phys. Lett. A {\bf 232} (1997) 333--339.
%arXiv:quant-ph/9703004


\bibitem{horo-survey}
R. Horodecki, P. Horodecki, M. Horodecki and K. Horodecki,
\it Quantum entanglement,
\rm Rev. Mod. Phys. \bf 81 \rm (2009), 865--942.
%arXiv:quant-ph/0702225

\bibitem{kay_2011}
A. Kay,
\it Optimal detection of entanglement in Greenberger-Horne-Zeilinger states,
\rm Phys. Rev. A {\bf 83} (2011), 020303(R).

\bibitem{kiem-kye-na}
Y.-H. Kiem, S.-H. Kye and J. Na,
\it Product vectors in the ranges of multi-partite states with positive partial transposes and permanents of matrices
\rm Commun. Math. Phys. {\bf 338} (2015), 621--639.

\bibitem{kirk}
K. A. Kirkpatrick,
\it Uniqueness of a convex sum of products of projectors,
\rm J. Math. Phys. {\bf 43} (2002), 684--686.

\bibitem{kye_3qb_EW}
S.-H. Kye,
\it Three-qubit entanglement witnesses with the full spanning properties
\rm J. Phys. A: Math. Theor. {\bf 48} (2015), 235303.
%link,  arXiv:1501.00768

\bibitem{mendo}
P. E. M. F. Mendonca, S. M. H. Rafsanjani, D. Galetti and M. A. Marchiolli,
\it Maximally genuine multipartite entangled mixed X-states of N-qubits,
\rm J. Phys. A: Math. Theor. {\bf 48} (2015), 215304.

\bibitem{peres}
A. Peres,
\it Separability Criterion for Density Matrices,
\rm Phys. Rev. Lett. \bf 77 \rm (1996), 1413--1415.
%(arXiv:quant-ph/9604005)

\bibitem{Rafsanjani}
S. M. H. Rafsanjani, M. Huber, C. J. Broadbent and J. H. Eberly
\it Genuinely multipartite concurrence of N-qubit X matrices,
\rm Phys. Rev. A {\bf 86} (2012), 062303.

\bibitem{rau}
A. R. P. Rau,
\it Algebraic characterization of X-states in quantum information,
\rm J. Phys. A: Math. Theor. {\bf 42} (2009), 412002.

\bibitem{vin10}
S. Vinjanampathy and A. R. P. Rau,
\it Generalized $X$ states of $N$ qubits and their symmetries,
\rm Phys. Rev. A {\bf 82} (2010), 032336.

\bibitem{wein10}
Y. S. Weinstein,
\it Entanglement dynamics in three-qubit $X$ states,
\rm Phys. Rev. A {\bf 82} (2010), 032326.

\bibitem{yu}
T. Yu and J. H. Eberly,
\it Quantum Open System Theory: Bipartite Aspects,
\rm Phys. Rev. Lett. {\bf 97} (2006), 140403.

\end{thebibliography}
\end{document}